\documentclass[11pt]{article}
\usepackage{amssymb,amsfonts,amsmath,amsthm,amscd,dsfont,mathrsfs}
\usepackage{graphicx,float,psfrag,epsfig,color,url}
\usepackage{latexsym}

\def\E{{\mathbb E}}
\def\ve{\varepsilon}

\def\cX{{\cal X}}
\def\cS{{\cal S}}

\def\E{\mathbb{E}}

\def\de{{\rm d}}
\def\ind{{\mathbb I}}
\def\ker{\text{ker}}

\newcommand{\effpee}{{\cal F}_p}
\newcommand{\ellpee}{\ell_p}

\newcommand{\hx}{\widehat{x}}
\newcommand{\halpha}{\widehat{\alpha}}
\newcommand{\hf}{\widehat{f}}

\newcommand{\exxohenn}{x_0^{(N)}}
\newcommand{\zeeenn}{z^{(n)}}
\newcommand{\ayyenn}{A^{(n,N)}}
\newcommand{\hxlenn}{\hat{x}_{\lambda}^{(N)}}
\newcommand{\bx}{{\bf x}}
\newcommand{\bz}{{\bf z}}

\newcommand{\normal}{\hat{x}_{1.\lambda}}
\def\normal{{\sf N}}

\newcommand{\enn}{{\sf n}}

\newcommand{\mse}{{\sf mse}}
\newcommand{\m}{{\sf mse}_0}

\def\de{{\rm d}}

\def\reals{{\mathbb R}}

\def\E{{\mathbb E}}
\def\<{\langle}
\def\>{\rangle}

\def\prob{{\mathbb P}}

\def\SE{\mbox{\tiny\rm SE}}

\def\cI{{\cal I}}
\def\Gauss{\mbox{\sc Gauss}}

\newcommand{\beqa}{\begin{eqnarray}}
\newcommand{\eeqa}{\end{eqnarray}}

\newcommand{\bitem}{\begin{itemize}}
\newcommand{\eitem}{\end{itemize}}
\newcommand{\beq}{\begin{equation}}
\newcommand{\eeq}{\end{equation}}
\newcommand{\goto}{\rightarrow}

\newcommand{\cA}{{\cal A}}
\newcommand{\cP}{{\cal P}}

\newcommand{\bR}{{\bf R}}

\newcommand{\cZ}{{\cal Z}}

\newcommand{\eps}{{\varepsilon}}

\definecolor{cardinal}{rgb}{.64,0.,.11}

\numberwithin{equation}{section}
\newtheorem{theorem}{Theorem}[section]
\newtheorem{corollary}{Corollary}[section]
\newtheorem{lemma}{Lemma}[section]
\newtheorem{proposition}{Proposition}[section]
\newtheorem{definition}{Definition}[section]

\def\MSE{{\rm MSE}}

\def\AMSE{{\rm AMSE}}

\def\HFP{{\rm HFP}}

\def\MSE{{\rm MSE}}

\def\LASSO{{\rm LASSO}}

\def\id{{\bf I}}
\def\Seq{{\sf S}}

\def\th{\tau}

\newcommand{\stMSE}{{\sf mse}}
\newcommand{\npi}{{\sf npi}}

\newcommand{\cF}{{\cal F}}

\title{Compressed Sensing
over $\ell_p$-balls:\\
Minimax Mean Square Error}
\author{David Donoho, Iain Johnstone, Arian Maleki and Andrea Montanari}

\begin{document}
\maketitle

\begin{abstract}
We consider the compressed sensing problem,
where the object $x_0 \in \bR^N$ is to be recovered
from incomplete measurements $y = Ax_0 + z$; here the
sensing matrix $A$ is an $n \times N$ random matrix with iid Gaussian
entries and $n < N$. A popular method of sparsity-promoting
reconstruction is $\ell^1$-penalized least-squares
reconstruction (aka LASSO, Basis Pursuit).

It is currently popular to consider the strict sparsity model,
where the object $x_0$ is nonzero in only a small fraction of entries.
In this paper, we instead consider  the much more broadly applicable $\ell_p$-sparsity model,
where $x_0$ is sparse in the 
sense of having $\ell_p$ norm bounded by $\xi \cdot N^{1/p}$
for some fixed $0 < p \leq 1$ and $\xi > 0$.  

We study an asymptotic regime in which $n$ and $N$ both tend to infinity
with limiting ratio $n/N = \delta \in (0,1)$,  both in the noisy ($z \neq 0$)
and noiseless ($z=0$) cases. Under weak assumptions on $x_0$, we are able to precisely evaluate the worst-case
asymptotic minimax mean-squared reconstruction
error (AMSE) for  $\ell^1$ penalized least-squares: min over penalization parameters,
max over $\ell_p$-sparse objects $x_0$. We exhibit
the asymptotically least-favorable object  (hardest sparse signal
to recover) and the maximin penalization.

In the case where $n/N$ tends to zero slowly -- i.e. extreme undersampling -- our formulas (normalized for comparison) say  that
the minimax AMSE of $\ell_1$ penalized least-squares is asymptotic to  $   \xi^2 \cdot (\frac{2\log(N/n)}{n})^{2/p-1}  \cdot (1 + o(1))$. Thus
we have not only the rate but also the  constant factor on the AMSE; and the maximin penalty factor
needed to attain this performance is also precisely specified.  Other similarly precise calculations are 
showcased.

Our explicit formulas unexpectedly involve
quantities appearing classically in statistical decision theory. %-- expressions for the minimax MSE
%of soft thresholding estimation of a single normal mean. 
Occurring in the present setting,
 they reflect a deeper connection between penalized $\ell^1$ 
 minimization and scalar soft thresholding.  This connection, which follows from
 earlier work of the authors and collaborators on the AMP iterative thresholding algorithm, 
 is carefully explained.

Our approach
also gives precise results under weak-$\ell_p$ ball
coefficient constraints, as we show here.

\end{abstract}

\vspace{.1in}
Lasso. Basis Pursuit. Minimax Risk over Nearly-Black Objects.
Minimax Risk of Soft Thresholding.
\vspace{.1in}

\vspace{.1in}
{\bf Acknowledgements.}
NSF DMS-0505303 \& 0906812, NSF CAREER CCF-0743978 .
\vspace{.1in}

\newpage

\section{Introduction}

In the compressed sensing problem,
we are given a collection of noisy, linear measurements of an unknown
vector $x_0$
\beq
\label{eq:obsdata} y=Ax_0 + z,
\eeq
Here the measurement matrix $A$ has dimensions $n$ by $N$, $n < N$,
the $N$-vector $x_0$ is the object we wish to recover
 and the noise $z\sim\normal(0,\sigma^2 \id )$.
Both $y$ and $A$ are known, both $x_0$ and $z$ are unknown, and
we seek an approximation to $x_0$.

Since the equations are underdetermined and noisy,
it seems hopeless to recover $x_0$ in general, but
in compressed sensing one also assumes that the object is
\emph{sparse}.  In a number of recent papers, the sparsity assumption 
is formalized by  requiring $x_0$ to have at most $k$ nonzero entries.
 This \emph{$k$-sparse model}
leads to a simpler analysis,  but is highly idealized,
and does not cover situations where a few dominant
entries are scattered among many small but slightly nonzero entries.
For such situations, \cite{Donoho1} proposed to measure
sparsity by membership in $\ell_p$ balls  $ 0 < p \leq 1$, namely 
to consider the situation
where  the $\ell_p$-norm\footnote{Throughout this paper we will 
accept the abuse of terminology of calling $\|\,\cdot\,\|_p$
a `norm', although it is not a norm for $p<1$.}
 of $x_0$ is bounded as 
\begin{equation} \label{def-ell-pee} 
\| x_0 \|_p^p\equiv\sum_{i=1}^p|x_{0,i}|^p \leq N\xi^p\, ,
\end{equation}
for some constraint parameter $\xi$.  Here, as $p \goto 0$, we recover 
the $k$-sparse case (aka $\ell_0$ constraint).

Much more is known today about behavior of reconstruction algorithms under the $k$-sparse
model than in the more realistic $\ell_p$ balls model.   
In some sense the $k$-sparse model
has been more amenable to precise analysis.  
In the noiseless setting,
precise asymptotic formulas are now known  for the sparsity level 
$k$ at which $\ell_1$ minimization fails to correctly recover
the object $x_0$  \cite{Do05,DoTa05,DoTa10}.
In the noisy setting, precise asymptotic formulas are now known for the 
worst-case asymptotic mean-squared error of reconstruction by  $\ell_1$-penalized $\ell_2$ minimization
\cite{NSPT,BM-MPCS-2010}.  
By comparison, existing results for the $\ell_p$ balls
model are mainly qualitative estimates, i.e. bounds that capture the correct 
scaling with the problem dimensions but involve loose or unspecified 
multiplicative coefficients. We refer to Section
\ref{sec:LiteratureSection} for a brief overview of this line of work,
and a comparison with our results.

We believe our paper brings the state of knowledge about the $\ell_p$-ball
sparsity model to the same level of precision as for the $k$-sparse model.
We consider here the high-dimensional setting $N,n\to\infty$
with matrices $A$ having iid Gaussian entries.
We treat both the noisy and noiseless cases in a unified formalism
and provide precise expressions, including constants, describing the
worst-case large-system behavior of mean-squared error for 
optimally-tuned $\ell^1$-penalized
reconstructions.    Because our expressions
are precise, they deserve close scrutiny;
as we show here, this attention is rewarded with surprising  insights, such as the
equivalence of undersampling with adding additional noise.
Less precise methods could not provide such insights.

The rest of this introduction reviews the 
results obtained through our method. 

\subsection{Problem formulation; Preview of Main Results}

Our main results concern $\ell^1$-penalized least-squares 
reconstruction with penalization parameter $\lambda$.
\begin{equation} \label{eq-ell-1-pen-least-squares}
    \hx_{\lambda}\equiv \arg\min_x \Big\{\frac{1}{2}
\|y - Ax\|_2^2  + \lambda \| x \|_1 \Big\}\, .
\end{equation}
This reconstruction rule became popular under the names of LASSO \cite{Tibs96}
or Basis Pursuit DeNoising \cite{BP95}.
Our analysis involves a large-system limit, which was
effectively also used in \cite{DMM09,NSPT,BM-MPCS-2010}.
We introduce some convenient terminology:

\begin{definition}
A {\bf problem instance} $I_{n,N}$ is a triple 
$I_{n,N} = (x_0^{(N)},z^{(n)},A^{(n,N)} )$ consisting of an object $\exxohenn$ to recover,
a noise vector $\zeeenn$, and a measurement matrix $A$.  
A {\bf sequence of instances} $\Seq = (I_{n,N})$ is an infinite
sequence of such problem instances.
\end{definition}

At this level of generality, a sequence of instances is nearly arbitrary.
We now make specific assumptions on the members of each triple.
HEre and below $\ind({\cal P})$ is the indicator function on property
${\cal P}$.
\begin{definition}
\bitem
\item {\bf Object $\ell_p$ sparsity constraint.}
A sequence $\bx_0 = (\exxohenn)$ belongs to  $\cX_p(\xi)$ if $(i)$
$N^{-1} \|x_0^{(N)}\|_p^p \leq \xi^p$, for amm $M$;
and $(ii)$ There exists a sequence $B= \{B_M\}_{M\ge 0}$ such that
$B_M\to 0$, and for every $N$,
$\sum_{i=1}^N(x_{0,i}^{(N)})^2\ind(|x_{0,i}^{(N)}|\ge M) \leq B_MN$.
\item {\bf  Noise power constraint.}
A sequence $\bz = (\zeeenn)_n$ belongs to  $\cZ^2(\sigma)$ if  $n^{-1} \|z^{(n)}\|_2^2 \goto  \sigma^2$.
\item {\bf Gaussian Measurement matrix.} $\ayyenn \sim \mbox{\sc Gauss}(n,N)$ is an  $n \times N$  random matrix
with entries drawn iid from the $\normal(0,\frac{1}{n})$ distribution.
\item The {\bf Standard $\ell_p$ Problem Suite} $\cS_p(\delta,\xi,\sigma)$ is the collection of sequences of instances  $\Seq =  \{ (x_0^{(N)},z^{(n)},A^{(n,N)} ) \}_{n,N}$ where \\ 
(i) $n/N \goto \delta$,  \\
(ii) $\bx_0 \in \cX_p(\xi)$,  \\  
(iii) $\bz \in \cZ^2(\sigma)$, and  \\ 
(iv) each $\ayyenn$ is sampled from the Gaussian ensemble $\Gauss(n,N)$.
\eitem
\end{definition}
The uniform intergrability condition
$\sum_{i=1}^N(x_{0,i}^{(N)})^2\ind(|x_{0,i}^{(N)}|\ge M) \leq B_MN$
essentially requires that the $\ell_2$ norm of $x_0^{(N)}$ is not
dominated by a small subset of entries. As we discuss below, it
is a fairly weak condition and most likely can be removed because 
the least-favorable vectors $x_0$ turn out to have all non-zero
entries of the same magnitude. Finally notice that 
 uniform integrability is implied by following: there exist $q>2$,
$B<\infty$ such that  $\|x_0^{(N)}\|_{q}^{q}\le N B$
for all $N$.

The fraction $\delta = n/N$ measures the incompleteness 
of the underlying systems of equations, with
$\delta$ near $1$ meaning $n \approx N$ and so
nearly complete sampling, and $\delta$  near $0$
meaning $n \ll N$ and so highly incomplete sampling.

Note in particular: the estimand $\bx$ and the noise $\bz$ are deterministic sequences
of objects, while the matrix $A$ is random.  In particular, while it may seem natural to pick
the noise to be random, that is not necessary, and in fact plays no role in our results.

Also let $\AMSE(\lambda; \Seq)$ denote the asymptotic per-coordinate 
mean-squared error of the LASSO reconstruction with penalty parameter $\lambda$,
for the sequence of problem instances $\Seq$ :
\begin{eqnarray}
\AMSE(\lambda,\Seq) = \limsup
\frac{1}{N}\, \E\big\{\|\hx_{\lambda}^{(N)}-\exxohenn\|^2\big\} \, .
\label{eq:AMSE_Def}
\end{eqnarray}
Here $\hx_{\lambda}^{(N)}$ denotes the LASSO estimator, and $x_0^{(N)}$ the estimand, on problem instances of
size\footnote{It would be more notationally correct to write
  $\hx_{\lambda}^{(N,n)}$ since the full problem size involves both
  $n$ and $N$, but we ordinarily have in mind a specific value $\delta
  \sim  n/N$, hence $n$ is not really free to vary independent of
  $N$.}
$N$.
Moreover the limsup is taken as $n,N\to\infty; n \sim \delta N$.  Although in general this quantity need not be well defined,
our results imply that, if the sequence of instances $\Seq$ is taken from the standard problem suite, this quantity
is bounded.

Now the AMSE depends on both $\lambda$, the penalization parameter,
and $\bx$, the sequence of objects  to recover. As in traditional
statistical decision theory, we may view the AMSE as the payoff function
of a game against Nature, where Nature chooses the object sequence $\bx$ 
and the researcher chooses the threshold parameter $\lambda$. 
In this paper, Nature is allowed to pick only sparse objects $\exxohenn$
obeying the constraint $ N^{-1} \| \exxohenn \|_p^p \leq \xi^p$.

In the case of noiseless information, $y = Ax_0$ (so $z = 0$),
this game has a saddlepoint, and Theorem \ref{thm:ellp:noiseless}  gives a precise evaluation of the minimax AMSE:
\begin{eqnarray}
\sup_{\Seq \in \cS_p(\delta,\xi,0)} \inf_\lambda \AMSE(\lambda, \Seq)  = \frac{\delta \xi^2}{M_p^{-1}(\delta)^2}\, .\label{eq:FirstNoiseless}
\end{eqnarray}
The maximin on the left side is  the payoff of a
zero-sum game. 
%Nature chooses a  distribution 
% $\nu$ obeying $E_\nu |X|^p \leq \xi^p$, and the researcher chooses
% a penalization $\lambda$; Nature makes the situation as difficult
% as possible for the researcher, who tries to control the damage that nature can cause.
%$\exxohenn$ is, for any $N$, within the $\ell_p$ ball of radius $\xi N^{1/p}$:
%%
%\begin{eqnarray}
%%
%\cS_p(\xi) \equiv\Big\{\Seq\, :\, \|\exxohenn\|_p^p\le N\xi^p
%\mbox{ for all } N \; ;\;
%x_{0,i}^{(N)}^2 \mbox{ uniformly  integrable }\Big\}\, .\label{eq:FirstNoiselessClass}
%%
%\end{eqnarray}
%%
%(The uniform integrability condition is mild but somewhat technical and 
%will be defined in Section \ref{sec:AMSEGeneral} below.) 

The function on the right side, $M_p(\,\cdot\,)$ is  displayed in Figure 1.
It evaluates the minimax MSE in a classical and much discussed problem of 
statistical decision theory: 
soft threshold estimation of random means $X$ satisfying the moment constraint 
$\E\{|X|^p\}\le \xi^p$ from noisy data $X + \normal(0,1)$. 
This problem was studied in \cite{DJ94a}, and detailed
information is known about $M_p$; see Section
\ref{sec:Scalar} for a review.

In the noisy case,  $\sigma > 0$, we have the same setup as before,
only now the AMSE will of course be larger.
Theorem \ref{thm:ellp:noisy} gives the minimax AMSE precisely:
% {thm:ellp:noiseless}
\begin{align}
\sup_{\Seq \in \cS_p(\delta,\xi,\sigma)} \inf_\lambda \AMSE(\lambda, \Seq) 
&=  \sigma^2 \cdot m_p^*(\delta,\xi/\sigma) \, ,
\label{eq:FirstNoisy}
\end{align}
where $m_p^*= m_p^*(\delta,\xi)$ is defined 
as the unique positive solution of the equation
\begin{align}
\frac{m}{1+m/\delta}  & =  M_p\left(\frac{\xi}
{(1+m/\delta)^{1/2}}\right)\, .\label{eq:SecondNoisy}
\end{align}
Again, the precise formula involves $M_p(\,\cdot\,)$, a classical
quantity in statistical decision theory.   See Figure 8 for a display of the minimax AMSE
as a function of $p$ and $\xi$.

Our results include several other precise formulas; our approach
is able to evaluate a number of operationally important
quantities
\bitem
 \item The {\it least-favorable} object, ie.
  the sparse estimand $x_0$ which causes
 maximal difficulty for the LASSO; Eqs (4.4), (5.5), (6.6).
 \item The {\it maximin tuning}, the actual choice of penalization
  which minimizes the AMSE when Nature chooses the least-favorable
  distribution; Eqs (4.3), (5.6), (6.16).
  \item Various operating characteristics, including the AMSE of reconstruction,
    and the limiting $\ell_p$ norms of the reconstruction.
 \eitem
 
 Various figures and tables present precise calculations
 which one can make using the results of this paper.  Figure $5$ shows
 the Minimax AMSE as a function of $\delta > 0$, for the noiseless case $z=0$ with fixed $\xi=1$,
 while Figure $8$ gives the minimax AMSE as a function of $\xi$ for fixed $\delta=1/4$,
 for the noisy case where the mean-square value of $z$ is $\sigma^2$.
 
\subsection{Novel Interpretations}

Our precise formulas provide not only accurate numerical
information, but also rather surprising insights. The appearance of 
the classical quantity $M_p$ in these formulas tells us that 
 a \emph{noiseless} compressed sensing problem, with {\it nonsquare}
sensing matrix $A$ having $n < N$ is explicitly connected 
 with the MSE in a very simple {\it noisy} problem
where $n=N$, $A$ is {\it square} -- in fact, the identity(!) --  cf. Eq.~(\ref{eq:FirstNoiseless}). 
 On the other hand,
 a \emph{noisy} compressed sensing problem with $n < N$ 
 and so $A$ nonsquare is explicitly connected 
 with a seemingly trivial problem,
where $n=N$ and $A$ is the identity, but the 
noise  level is {\it different} than in the compressed sensing problem -- in fact  {\it higher} --
cf. Eqs.~(\ref{eq:FirstNoisy}), (\ref{eq:SecondNoisy}).
Conclusion:
\begin{quote}
\emph{{\bf Slogan:} In both the noisy and noiseless cases: undersampling is effectively equivalent  to adding noise
to complete observations.}\footnote{The formal equivalence of undersampling to simply adding noise is 
quite striking. It reminds us of ideas from the so-called comparison
of experiments in traditional statistical decision theory.}
\end{quote} 
While \cite{StOMP} and \cite{CSMRI} formulate
heuristics  and provided empirical evidence about this connection,
the results here (and in the companion papers \cite{DMM09,NSPT}) provide the
only theoretical derivation of such a connection.

%Notice that in (\ref{eq:SecondNoisy}), unlike in the noiseless case, the radius of the 
%mean $\ell_p$ constraint on the signals $x_0$ (i.e. $\xi$)
%is larger than the moment bound appearing on the
%right-hand side (namely $\xi/(1+m_p^*/\delta)^{1/2}$).
%Conclusion:
%\begin{quotation}
%\emph{ {\bf Slogan:} In the noisy case: undersampling is effectively equivalent to an increase in
%the $\ell_p$ radius.}
%\end{quotation}

Established research
tools for understanding compressed sensing \-- for example estimates 
based on the restricted isometry property \cite{CandesTao,CandesStable} \--
provide upper bounds on the mean square error but do not allow one to 
suspect that such striking connections hold.  
In fact we use a very different approach
from the usual compressed sensing literature. Our methods
join ideas from belief propagation message passing in 
information theory,
and minimax decision theory in mathematical statistics.

\subsection{Complements and Extensions}

\subsubsection{Weak $\ell_p$}

Section 6  develops analogous results for  
compressed sensing  in the weak-$\ell_p$ balls model,  where the object
obeys a \emph{weak}-$\ell_p$ rather than an $\ell_p$ constraint.
Weak-$\ell_p$ balls are relevant models for natural images 
and hence our results have applications in image reconstruction, 
as we describe in Section 9. 

\subsubsection{Reformulation of $\ell_p$ Balls}
Our normalization of the error measure and
of $\ell_p$ balls are  somewhat different
than what has been called the $\ell_p$ case in earlier literature.
We also impose a tightness condition not present in earlier work.
In exchange, we get precise results.
For calibration of these results see Section 7.
From the practical
point of view of obtaining {\it accurate predictions about the behavior of real systems},
the present model has significant advantages.
For more detail, see Section \ref{sec-discuss}.

%Earlier literature on the $\ell_p$ balls model assumes a 
%single problem instance at a  fixed  $n$ and $N$ and takes a supremum
%across all $\exxohenn$ obeying a given $\ell_p$ ball constraint as that fixed $N$.
%
%We consider instead sequences of problem instances at increasing sizes.
%evaluate the $\limsup$ across problem sizes, 
%then consider the supremum over such instance sequences.  Roughly speaking, we first increase $N$
%and only later take the supremum; the traditional literature first takes the supremum and then increases $N$.
%
%Any deterministic vector $x_0$
%obeying the $\ell_p$ balls model at fixed, finite $N$ can be viewed as member of a sequence
%satisfying our assumptions; and so it is largely an issue of mathematical taste whether
%the original $\ell_p$ model  is to be preferred to the present model. 

%
%*************************************************************
%
\section{Minimax Mean Squared Error of Soft Thresholding}
\label{sec:Scalar}

%Our main results imply a connection between the 
%mean square error in a compressed sensing problem, and the
%mean square error in an `equivalent' scalar denoising problem.
%It is therefore convenient to summarize some basic facts about the latter.

Consider a signal  $x_0\in\reals^N$, and suppose
that it satsifies  $x_0$ satisfies the $\ell_2$-normalization $N^{-1}\|x_0\|_p^p \approx 1$ but also
 the $\ell_p$-constraint $\|x_0\|_p^p\le N \cdot \xi^p$, for small $\xi$ and $0 < p < 2$.
To see that this is a sparsity constraint,
note that a typical `dense' sequence, such as an iid Gaussian sequence,
cannot obey such a constraint for large $N$; in effect,
smallness of $\xi$ rules out sequences which have too many significantly nonzero values.

If we observed such a sparse sequence in additive Gaussian noise $y = x_0 + z$,
where $z \sim_{iid} \normal(0,1)$,  it is well-known that we could approximately recover the
vector by simple thresholding -- effectively, zeroing out the entries which are already close to zero.
  Consider  the soft-thresholding
nonlinearity $\eta:\reals\times\reals_+\to\reals$.
Given an observation $y\in\reals$ and a `threshold level' 
$\th\in\reals_+$, soft thresholding acts on a scalar as follows
\begin{eqnarray}
\eta(y;\th) = \left\{
\begin{array}{ll}
y-\th & \mbox{ if $y\ge \th$,}\\
0 & \mbox{ if $-\theta< y< \th$,}\\
y+\th  & \mbox{ if $ y\le -\th$.}
\end{array}
\right.
\end{eqnarray}
We apply it to a vector $y$ coordinatewise and 
get the estimate $\hat{x} = \eta(y;\tau)$.

To analyze this procedure we can work in terms of scalar random variables.
The empirical distribution of $x_0$ is defined as
\begin{eqnarray}
\nu_{x_0,N} \equiv\frac{1}{N}\sum_{i=1}^N\delta_{x_{0,i}}\, . \label{ecdf-def}
\end{eqnarray}
Define the random variables
 $X \sim \nu_{x_0,N}$ and $Z \sim \normal(0,1)$, with $X$ 
and $Z$ mutually independent.
We have the isometry:
\[
     N^{-1} \E \| \hat{x} - x_0 \|_2^2 = \E_{\nu_{x_0}} \big\{\big[\eta(X+Z;\tau) - X\big]^2\big\} .
\]
Hence, to analyze the behavior of thresholding under sparsity constraints,
we can shift attention from sequences in $\reals^N$ to distributions.

So define the class of `sparse' probability 
distributions over $\reals$:
\begin{eqnarray}
\cF_p(\xi) \equiv\Big\{\nu\in\cP(\reals)\, :\, \nu(|X|^p)\le \xi^p\Big\}\, ,
\label{eq:SparseProbs}
\end{eqnarray}
where $\cP({\reals})$ denotes the space of probability measures over
the real line.
Then $x_0$ satisfies the $\ell_p$-constraint $\|x_0\|_p^p\le N \cdot \xi^p$
if and only if $\nu_{x_0}\in\cF_p(\xi)$. 
 
The central quantity for our formulae
(\ref{eq:FirstNoiseless}), (\ref{eq:FirstNoisy})
is the minimax mean square error $M_p(\xi)$ defined now:
\begin{definition}
The \emph{minimax mean squared error} of soft thresholding
is defined by:
\begin{eqnarray}
M_p(\xi) = \inf_{\th\in\reals_+} \sup_{\nu\in\effpee(\xi)} 
\E\big\{ \big[ \eta(X + Z;\th) - X \big]^2 \big\}\, ,\label{eq:MpDef}
\end{eqnarray}
where expectation on the right hand side is taken with respect to
$X\sim \nu$ and $Z\sim\normal(0,1)$ mutually independent.
\end{definition}
This quantity has been carefully studied in \cite{DJ94a},
particularly in the asymptotic regime $\xi \goto 0$. Figure \ref{fig:emmPee}
displays its behavior as a function of $\xi$ for several different 
values of $p$.
\begin{figure}
\begin{center}
\includegraphics[height=3in]{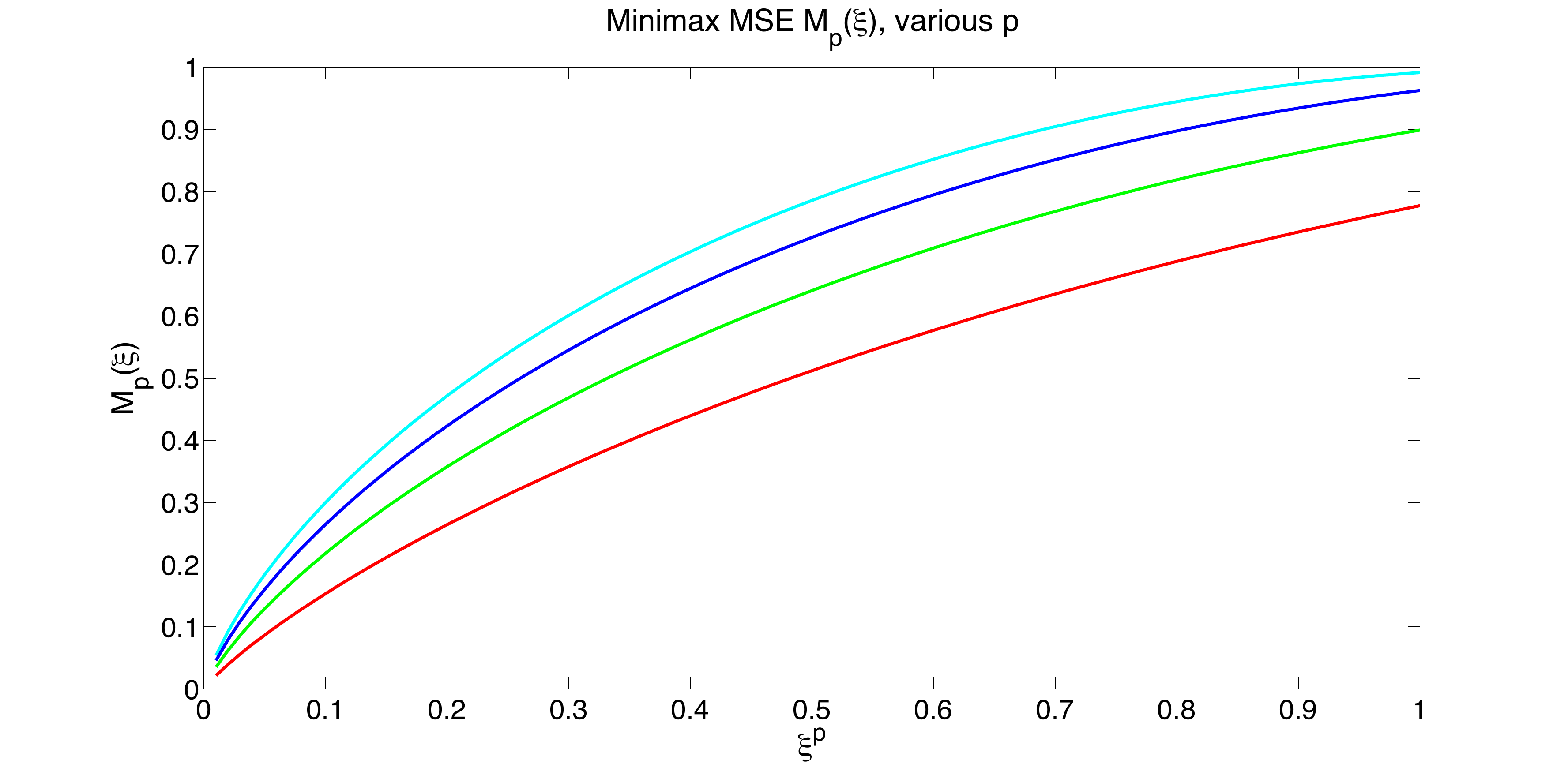}
\caption{Minimax soft thresholding risk, $M_p(\xi)$, various $p$.
Vertical axis: worst case MSE over $\cF_p(\xi)$. Horizontal axis: $\xi^p$. 
Red, green, blue, aqua curves correspond to $p= 0.1,0.25,0.50,1.00$.}
\label{fig:emmPee}
\end{center}
\end{figure}

The quantity (\ref{eq:MpDef}) can be viewed as the value
of a game against Nature,
where the statistician chooses the threshold $\th$,
Nature chooses the distribution $\nu$, and the statistician pays Nature 
an amount equal to the MSE.  We use
the following notation for the MSE of soft thresholding,
given a noise level $\sigma$, a signal distribution $\nu$ 
and a threshold level $\th$:
\begin{eqnarray}
\stMSE(\sigma^2;\nu,\th)\equiv
\E\big\{ \big[ \eta(X + \sigma\,Z;\th\,\sigma) - X \big]^2 \big\}\, ,
\label{eq:ScalarMSEDef}
\end{eqnarray}
where, again, expectation is with respect to
$X\sim \nu$ and $Z\sim\normal(0,1)$ independent.
Hence the quantity on the right hand side of Eq.~(\ref{eq:MpDef})
--the game payoff-- is just $\mse(1;\nu,\th)$.

Evaluating the supremum  in 
Eq.~(\ref{eq:MpDef}) might at first appear hopeless. In reality 
the computation can be done rather explicitly 
using the following result.
%{\bf [PROOF MISSING]}
%
\begin{lemma}\label{lemma:ScalarSaddle}
The least-favorable distribution $\nu_{p,\xi}$, i.e. the distribution forcing attainment of the worst-case MSE,
is supported on 3 points.
Explicitly, consider the 3-point mixture distribution
\begin{eqnarray}
\nu_{\eps,\mu} = (1-\eps) \delta_0 + \frac{\eps}{2} \delta_{\mu} + \frac{\eps}{2} \delta_{-\mu} \, .\label{eq:3-point}
\end{eqnarray}
Then the least-favorable distribution 
 $\nu_{p,\xi} $ is the 3-point mixture $\nu_{\eps_p(\xi),\mu_p(\xi)}$
for specific values $\eps_p(\xi),\mu_p(\xi)$. 
\end{lemma}
In fact it seems the minimax problem in Eq.~(\ref{eq:MpDef}) has a saddlepoint,
i.e.  a pair $(\nu_{p,\xi},\th_p(\xi))\in \cP(\reals)\times  \reals_+$,
such that
\begin{eqnarray}
 \mse(1;\nu_{p,\xi},\th) \geq \mse(1;\nu_{p,\xi},\th_p(\xi)) 
\geq \mse(1;\nu,\th_p(\xi)) \qquad \forall  \th > 0, \;  \forall \nu 
\in \effpee(\xi) \, ,
\end{eqnarray}
but we do not need or prove this fact here.
The MSE is readily evaluated for 3-point distribution,
yielding
\begin{align}
& \mse(1;\nu_{\eps,\mu},\th) = (1-\ve)\big\{
2(1+\th^2)\Phi(-\th)-2\th\phi(\th)\big\} \label{eq:ScalarMSE}\\
&+\ve
\big\{\mu^2+(1+\th^2-\mu^2)[\Phi(-\mu-\th)+\Phi(\mu-\th)]
+(\mu-\th)\phi(\mu+\th)-(\mu+\th)\phi(-\mu+\th)\big\}\,.\nonumber
\end{align}
Here and below, $\phi(z) \equiv e^{-z^2/2}/\sqrt{2\pi}$ is the 
standard Gaussian density and $\Phi(x) \equiv \int_{-\infty}^x\phi(z)\,\de z$
is the Gaussian distribution function. 
Further, it is easy to check that the MSE is maximized
when the $\ell_p$ constraint is saturated, i.e. for
\begin{eqnarray}
\eps\mu^p = \xi^p\, .
\end{eqnarray}
Therefore one is left with the task of maximizing the 
right-hand side of Eq.~(\ref{eq:ScalarMSE}) with respect to $\ve$
(for $\mu = \xi\ve^{-1/p}$) and minimizing it with respect to $\th$.
This can be done quite easily numerically for any given $\xi>0$,
yielding  the values of $\tau_p(\xi)$, $\mu_p(\xi)$
and $\eps_p(\xi)$ plotted in Fig.~\ref{fig:LFMu}.
The minimax property is illustrated in Fig.~\ref{fig:MMThresh}.

\begin{figure}
\begin{center}
\includegraphics[height=2.0in]{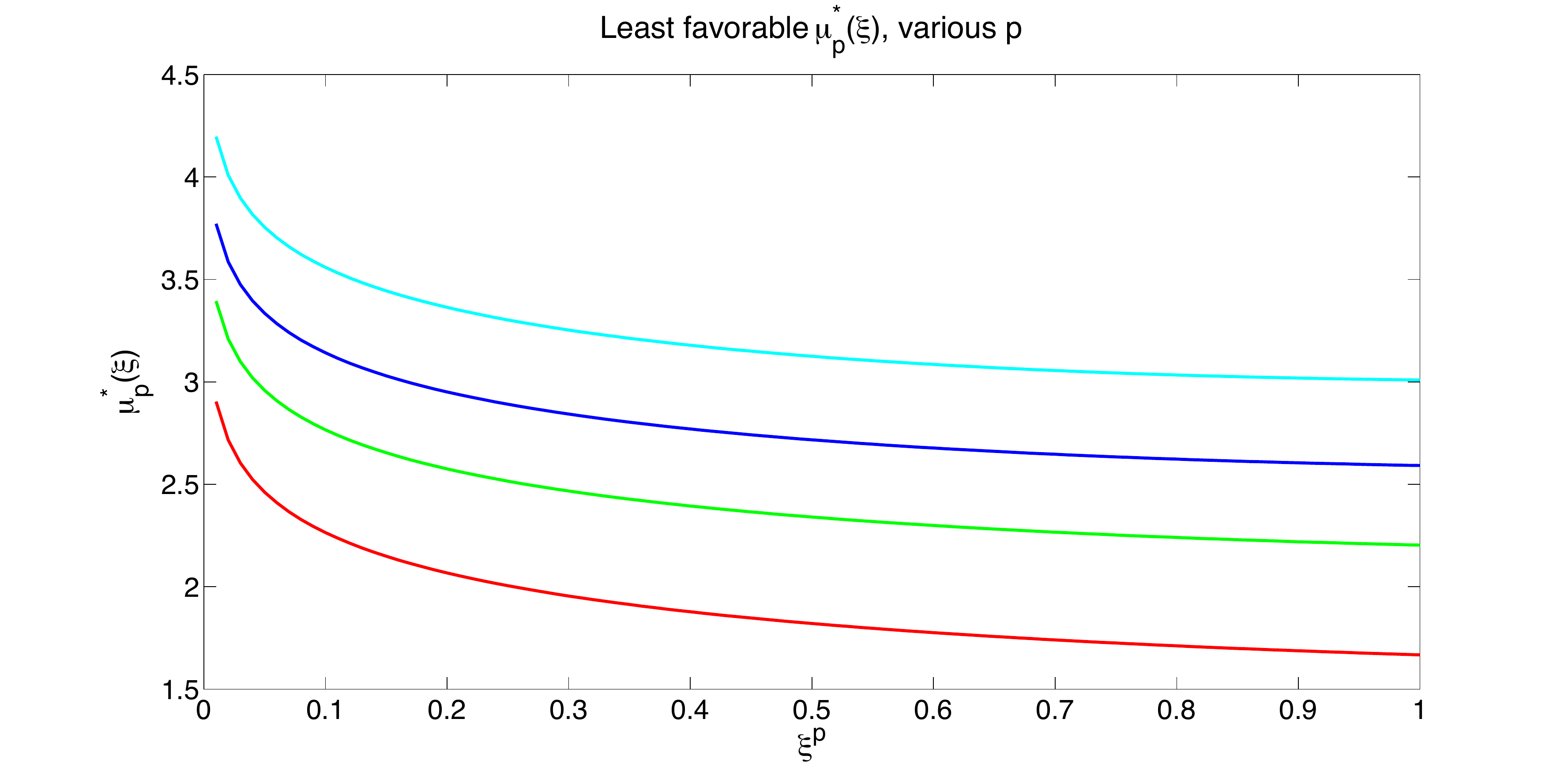} \quad 
\includegraphics[height=2.0in]{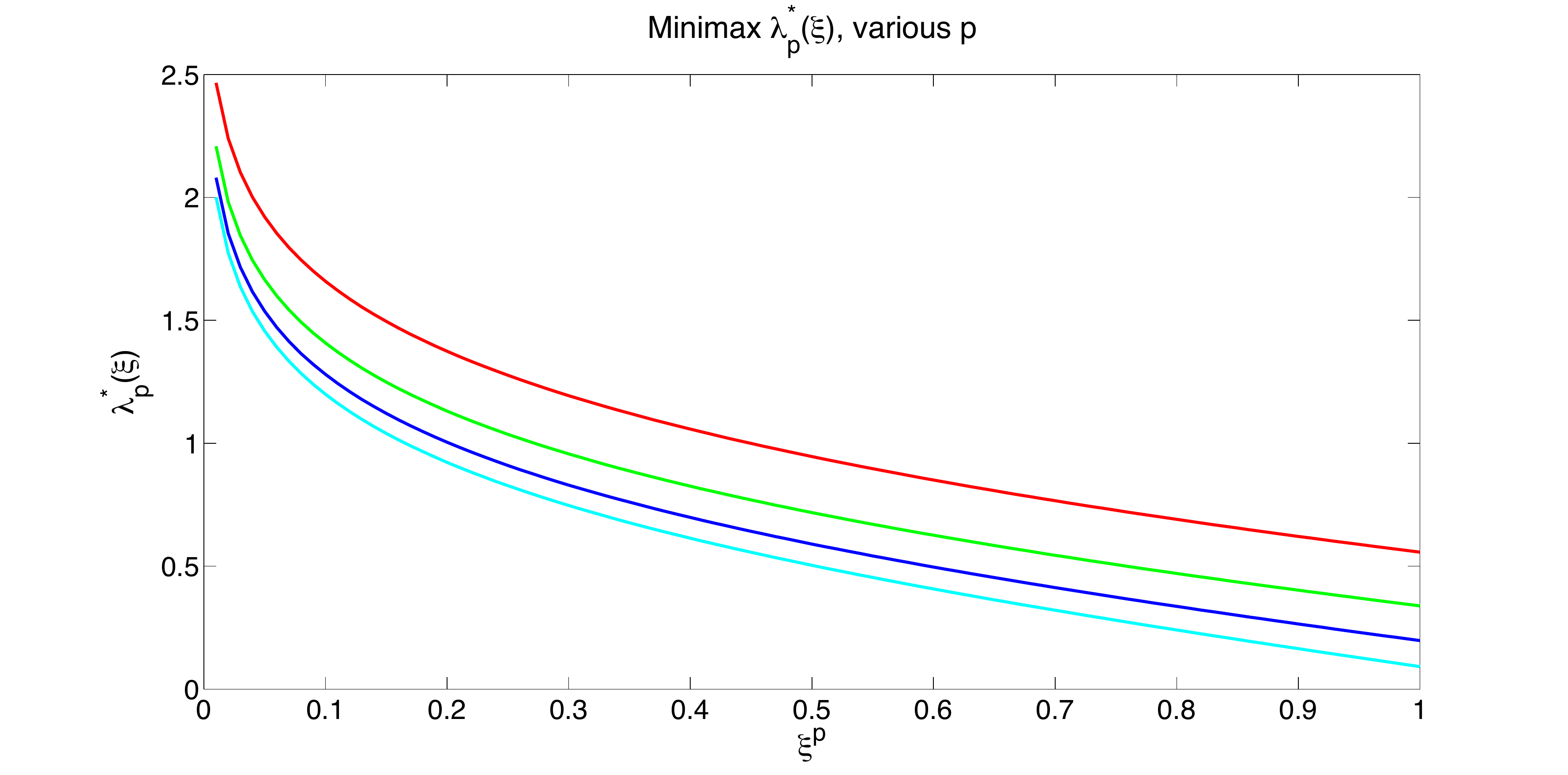}
\caption{Least-favorable $\mu$ (upper frame)
and corresponding minimax threshold $\th$ (lower frame).
Horizontal axes: $\xi^p$. 
Red, green, blue, aqua curves correspond to $p= 0.1,0.25,0.50,1.00$.}
%{\bf [labels are too small.]}}
\label{fig:LFMu}
\end{center}
\end{figure}

\begin{figure}
\begin{center}
\includegraphics[height=4in]{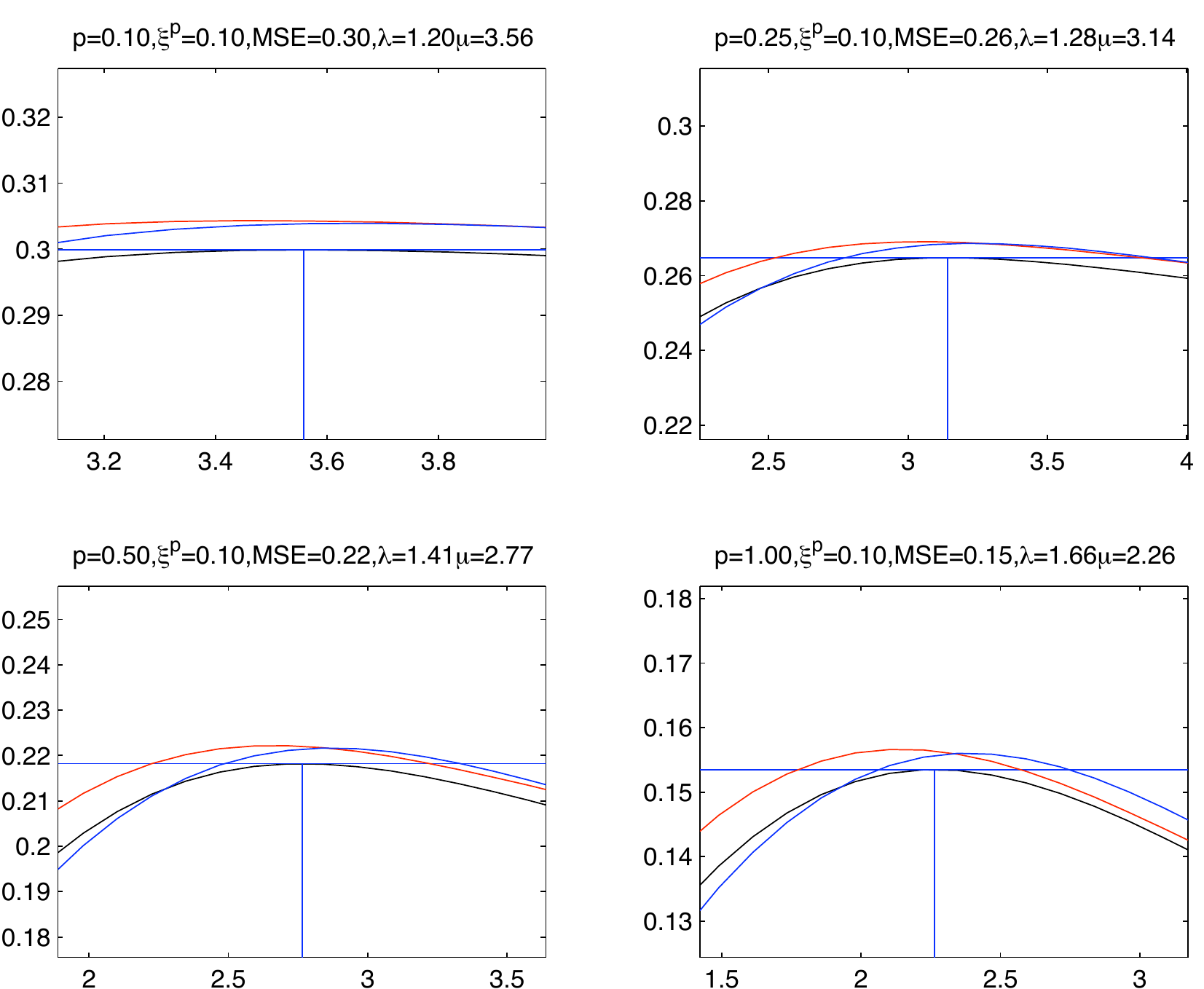}
\caption{Saddlepoint property of Minimax $\th_p(\xi)$, $\xi^p=1/10$, various $p$.
Vertical Axis: MSE at $F_{\eps,\mu}$. Horizontal Axis $\mu$.
Vertical Blue line: least-favorable $\mu$, $\mu_p(\xi)$.
Horizontal Blue Line: Minimax MSE $M_p(\xi)$.
At each value of $\mu$,  Black curve displays corresponding 
MSE of soft thresholding with threshold
at the minimax threshold value $\tau_p(\xi)$ ,
under the distribution $F_{\eps,\mu}$ with $\eps \mu^p = \xi^p$.
The other two curves are for $\tau$ 10 percent higher and 10 percent lower than
the minimax value.  In each case, the black curve (associated with minimax $\tau$),
stays below the horizontal line, while the red and blue curves cross above it,
illustrating the saddlepoint relation.
}\label{fig:MMThresh}
\end{center}
\end{figure}

Important below will be the inverse function
\begin{eqnarray}
M_p^{-1}(m) = \inf \big\{\,  \xi\in [0,\infty)\, : \;M_p(\xi) \geq m 
\, \big\} ,
\end{eqnarray}
defined for $m \in (0,1)$, and depicted in Figure \ref{fig:emmPeeInv}.
The well-definedness of this function follows from the next Lemma.
\begin{lemma} \label{lem:def:inv:mp}
The function
$\xi\mapsto M_p(\xi)$ is continuous and strictly increasing for
$\xi \in (0,\infty)$, with $\lim_{\xi\to 0}M_p(\xi)=0$,
and $\lim_{\xi\to \infty}M_p(\xi)=1$.
\end{lemma}
\begin{proof}
Let $\m( \mu,\tau) \equiv \E\{[\eta(\mu+Z;\tau)-\mu]^2\}$
for $Z\sim\normal(0,1)$, so that
$\stMSE(1;\tau, \nu) = \int \m(\mu,\th)\, \nu( \de\mu)$.
Since  $\m(\mu,\th) = \m(-\mu,\th)$ 
in this formula we can assume without loss of generality that 
$\nu(\,\cdot\,)$ is supported on $\reals_+$.

To show strict monotonicity, fix $\xi \leq \xi'$, let $\tau' =
\tau_p(\xi')$ be the minimax threshold for $\cF_p(\xi')$, and let
$\nu_\xi = \nu_{p,\xi}$ be the least favorable prior for
$\cF_p(\xi)$. 
Let $\nu' = S_{\xi'/\xi} \nu_\xi$ be the measure in $\cF_p(\xi')$
obtained by scaling $\nu_\xi$ up by a factor $\xi'/\xi$
(explicitly, for a measurable set $C$, $\nu'(C)= \nu_{\xi}((\xi'/\xi)C)$). 
Since $\nu_{\xi}\neq \delta_0$,
strict monotonicity of $\mu \rightarrow \stMSE_0(\mu,\tau)$ 
(e.g. \cite[eq. A2.8]{DJ94a}) shows that 
$\stMSE(1; \tau', \nu_\xi) < \stMSE(1; \tau', \nu')$.
Consequently
\begin{eqnarray*}
  M_p(\xi) \leq \stMSE (1; \tau', \nu_\xi) 
           < \stMSE(1; \tau', \nu')
             \leq \sup_{ \nu\in\cF_p(\xi')} \stMSE(1; \tau', \nu) = M_p(\xi').
\end{eqnarray*}

We verify that $t \rightarrow M_p(t^{1/p})$ is concave in $t$: 
combined with strict monotonicity, we can then conclude that $M_p(\xi)$ is
continuous. Indeed, the map $\nu \rightarrow \stMSE(1;\tau, \nu)$ 
is linear in $\nu$ and so 
$\stMSE_*(\nu) = \inf_\tau \stMSE(1;\tau, \nu)$ is concave in $\nu$.
Hence
$ M_p(t^{1/p}) = \sup \{ \stMSE_* (\nu) ~:~ \nu(|X|^p) \leq t \}$
is also concave. 

That $\lim_{\xi \rightarrow 0} M_p(\xi) = 0$ is shown in \cite{DJ94a},
compare Lemma \ref{lem:asymp:scalr:minimax} 
below. For large $\xi$, observe that
\begin{displaymath}
  1 \geq M_p(\xi) \geq \mathcal{M}_p(\xi)
     \equiv \inf_\eta \sup_{\cF_p(\xi)} \E \{[ \eta(X+Z) - X ]^2\}\, ,
\end{displaymath}
the minimax risk over \textit{all} estimators $\eta$. Further 
$\mathcal{M}_p(\xi) \geq \mathcal{M}_\infty(\xi)$, the minimax risk
for estimation subject to the bounded mean constraint $|\mu| \leq
\xi$. 
That $\mathcal{M}_\infty(\xi) \rightarrow 1$ is shown, for example, in
\cite[Eq. (2.6)]{dlm90}.
\end{proof}

\begin{figure}
\begin{center}

\includegraphics[height=2.0in]{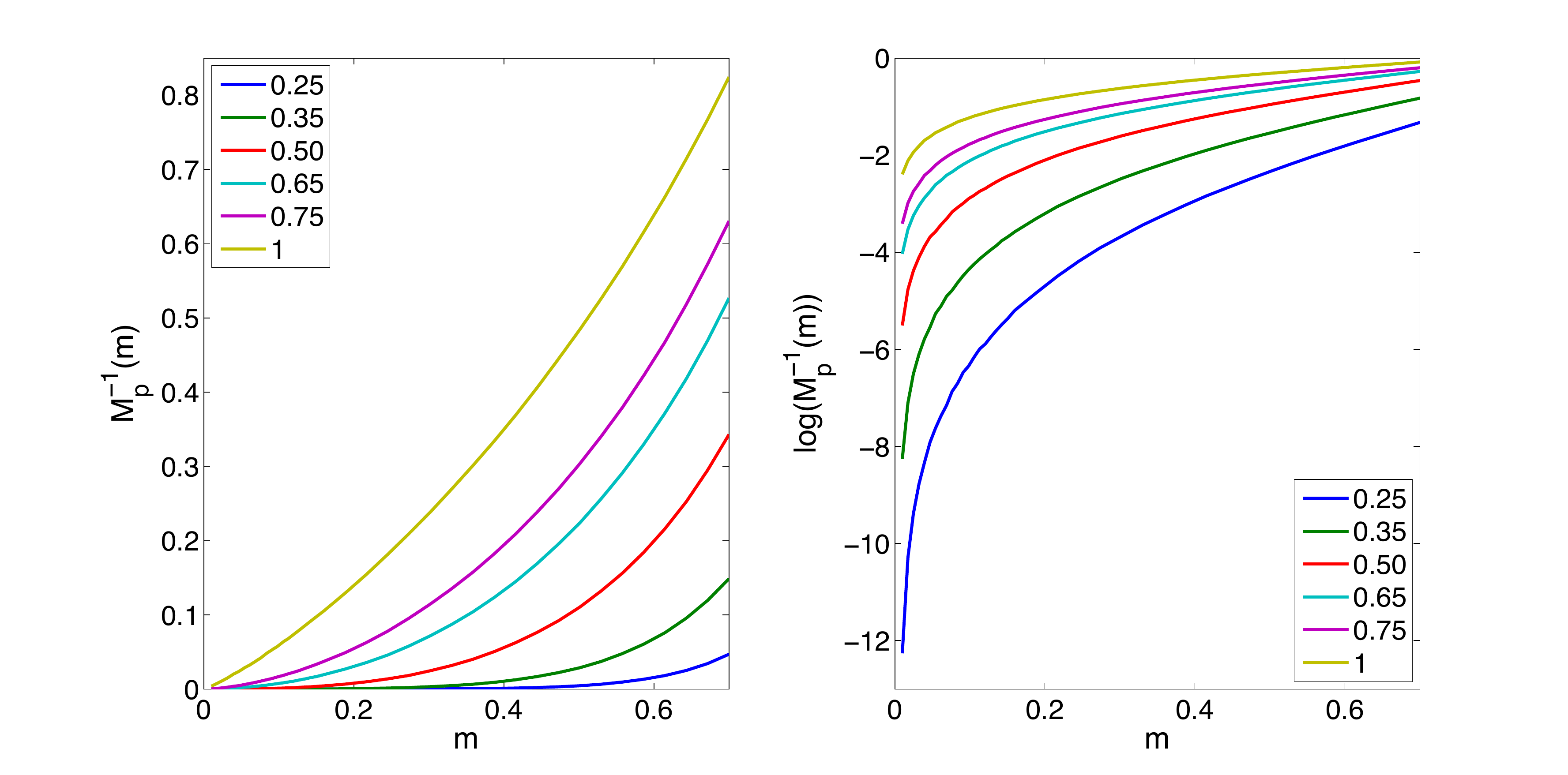}    %\;\;
\caption{Inverse function $M_p^{-1}(m)$.
Horizontal axis: $m$, desired minimax mean square error $m$. 
Vertical axis:  left-hand plot: $\xi = M_p^{-1}(m)$, the radius of ball that attains it.
right-hand plot: $\log(\xi)$.
Colored curves correspond to various choices of $p$.}
\label{fig:emmPeeInv}
\end{center}
\end{figure}

Of particular interest is the case of extremely sparse 
signals, which corresponds to the limit of small $\xi$.
This regime was studied in detail in \cite{DJ94a} 
whose results we summarize below.
\begin{lemma}[\cite{DJ94a}]\label{lem:asymp:scalr:minimax}
As $\xi \to 0$ the minimax pair
$(\nu_{\eps_p(\xi),\mu_p(\xi)},\th_{p}(\xi))$ in Eq.~(\ref{eq:MpDef}) 
obeys
\begin{eqnarray*}
\th_p(\xi) &=& \sqrt{2 \log( 1/\xi^{p})}\cdot\{1+o(1)\}\, , \\
\mu_p(\xi) & = & \sqrt{2 \log( 1/\xi^{p})}\, \cdot\{1+o(1)\}\,  , \\
\eps_p(\xi)  &=& \left(\frac{\xi^2}{2 \log( 1/\xi^{p})}\right)^{p/2} 
\cdot \{1+o(1)\}\, .
\end{eqnarray*}
Further, the minimax mean square error is given, in the same limit,
by
\beq \label{eq:asymp:MSE}
M_p(\xi) = (2 \log( 1/\xi^{p}))^{1-p/2} \xi^p\cdot \{1+o(1)\}\, .
\eeq
\end{lemma}

The asymptotics for $M_p(\xi)$ in the last 
lemma imply the following behavior
of the inverse function as $m\to 0$:
\beq \label{eq:asymp:MpInv}
     M_p^{-1}(m) = \big(2 \log(1/m)\Big)^{1/2-1/p} m^{1/p}
\cdot \{1+o(1)\}\,.
\eeq

%
%******************************************************
%
\section{The asymptotic LASSO risk }

In this section we discuss the high-dimensional limit of the 
LASSO mean square error  for a given sequence
of instances $\Seq = (I_{n,N})$. Our treatment is mainly a summary
of results proved in \cite{BayatiMontanariLASSO} and \cite{NSPT},
adapted to the current context.

%
%*************************************************************
%
\subsection{Convergent Sequences, and their AMSE}

We introduced the notion of sequence of instances as a very general,
almost structure-free notion;  but certain special sequences play a distinguished role.
\begin{definition}\label{def:Converging} {\bf Convergent sequence of problem instances.}
The sequence of problem instances  $\Seq = \{(\exxohenn, \zeeenn, \ayyenn)\}_{n,N}$
 is said to be a \emph{convergent sequence} if  $n/N\to\delta\in(0,\infty)$,
and in addition the following conditions hold:
\begin{itemize}
\item[$(a)$] {\bf Convergence of object marginals.} The empirical distribution of the entries of $\exxohenn$
converges weakly to a probability measure $\nu$ on $\reals$
with bounded second moment. Further
$N^{-1}\|x_{0}^{(N)}\|_2^2\to \E_{\nu}X^{2}$.
\item[$(b)$] {\bf Convergence of noise marginals.}  The empirical distribution of the entries of $\zeeenn$
converges weakly to a probability measure $\omega$ on $\reals$
with bounded second moment. Further
$n^{-1}\|z_{i}^{(n)}\|_2^2\to \E_{\omega}Z^{2}\equiv\sigma^2$.
\item[$(c)$] {\bf Normalization of Matrix Columns.} If $\{e_i\}_{1\le i\le N}$, $e_i\in\reals^N$ denotes the standard
basis, then $\max_{i\in [N]}\|\ayyenn e_i\|_2$, 
$\min_{i\in [N]}\|\ayyenn e_i\|_2\to 1$,
as $N\to\infty$ where $[N]\equiv\{1,2,\ldots,N\}$.
\end{itemize}
We shall say that $\Seq$ is a \emph{ convergent sequence of problem instances},
and will write $\Seq \in CS(\delta,\nu,\omega,\sigma)$ to make explicit the limit objects. 
\end{definition}
Next we need to introduce or  recall some notations. The mean square error for scalar soft thresholding 
was already introduced in the previous Section,
cf. Eq.~(\ref{eq:ScalarMSEDef}), and denoted by $\mse(\sigma^2;\nu,\th)$.
The second is the following \emph{state evolution map}
\begin{eqnarray}
\Psi(m;\delta,\sigma,\nu,\th)
\equiv  \stMSE\Big(\sigma^2+\frac{1}{\delta}m; \nu,\th\Big)\, ,\label{eq:PsiDef}
\end{eqnarray}
This is the mean square error for soft thresholding, 
when the noise variance is $\sigma^2+m/\delta$. The addition
of the last term reflects the increase of `effective noise'
in compressed sensing  as compared to simple denoising, due to the
undersampling. In order to have a shorthand for the latter,
we define \emph{noise plus interference} to be
\begin{eqnarray}
\npi(m;\delta,\sigma) = \sigma^2+\frac{m}{\delta}\, .
\end{eqnarray}
Whenever the arguments $\delta,\sigma,\nu,\th$ will be clear from the context
in the above functions, we will drop them and write, with an abuse of notation
$\Psi(m)$ and $\npi(m)$.

Finally, we need to introduce the following 
\emph{calibration relation}. Given $\th\in\reals_+$,
let $m_*(\th)$ to be the largest positive solution 
of the fixed point equation 
\begin{eqnarray}
m = \Psi(m;\delta,\sigma,\nu,\th)
\end{eqnarray}
(of course $m_*$ depends on $\delta,\sigma,\nu$ as well
but we'll drop this dependence unless necessary).
Such a solution is finite for all $\th> \th_{0}$ for 
some $\th_0 = \th_0(\delta)$.
The corresponding LASSO parameter is then given by
\begin{eqnarray}
\lambda(\th) \equiv \th \sqrt{\npi_*} \left[1 - \frac{1}{\delta}
\prob\big\{|X+\sqrt{\npi_*} Z|\ge \th
\sqrt{\npi_*}\big\}\right]\, .
\label{eq:Calibration}
\end{eqnarray}
with $\npi_*= \npi(m_*(\th))$.
As shown in \cite{BayatiMontanariLASSO}, $\th\mapsto\lambda(\th)$
establishes a bijection between $\lambda\in (0,\infty)$ and
$\th\in(\th_1,\infty)$ for 
some $\th_1 = \th_1(\delta)>\th_0(\delta)$.

The basic high-dimensional limit result can be stated as follows.
\begin{theorem}\label{thm:Risk}
Let $\Seq = \{ I_{n,N} \} = \{ (\exxohenn, \zeeenn, \ayyenn)\}_{n,N}$ be a convergent sequence of problem
instances, $\Seq \in CS(\delta,\sigma,\nu,\omega)$, and 
assume also that the matrices $\ayyenn$ are sampled from $\Gauss(n,N)$. 
Denote by $\hxlenn$ the \LASSO\, estimator for
instance $I_{n,N}$, $\lambda\geq 0$ and
let $\psi:\reals\times\reals\to \reals$ be a locally-Lipschitz function
with $|\psi(x_1,x_2)|\le C(1+x_1^2+x_2^2)$ for all $x_1,x_2\in\reals$.

Then, almost surely
\begin{eqnarray}\label{eq:asymptotic-result}
\lim_{N\to\infty}\frac{1}{N}\sum_{i=1}^N\psi
\big(\hx_{\lambda,i},x_{0,i}\big) = \E\Big\{\psi\big(\eta(X+
\sqrt{\npi_*}\, Z;\th_*\sqrt{\npi_*}),X\big)\Big\}\, ,
\end{eqnarray}
where $\npi_*\equiv \npi(m_*)$, 
$Z\sim\normal(0,1)$ is independent of $X\sim \nu$,
$\th_*=\th_*(\lambda)$ is given by the \emph{calibration
relation} described above, and $m_*$ is the largest positive solution 
of the fixed point equation $m = \Psi(m,\delta,\sigma,\nu,\th_*)$.
\end{theorem}
%
%****************************************************************
%
\subsection{Discussion and further properties}

In the next pages we will repeatedly use the shorthand
$\HFP(\Psi)$ to denote the largest positive solution of the fixed point
equation $m = \Psi(m;\delta,\sigma,\nu,\th)$, where we may suppress the
secondary parameters $(\delta,\sigma,\nu,\th)$ and simply write $\Psi(m)$.
Formally
\begin{eqnarray}
    \HFP(\Psi) \equiv \sup\{ m\ge 0\, : \Psi(m) \geq m \}.
\end{eqnarray}
In order to emphasize the role of parameters
$\delta,\sigma,\nu,\th$, we may also write\\
$\HFP(\Psi(\,\cdot\, ; \delta,\sigma,\nu,\th))$.
We recall some basic properties of the mapping
$\Psi$.
\begin{lemma}[\cite{DMM09,NSPT}]\label{lemma:Map}
For fixed $\delta,\sigma,\nu,\th$, the mapping
$m\mapsto \Psi(m)$ defined on 
$[0,\infty)$ is continuous, strictly increasing and concave.
Further $\Psi(0)\ge 0$ with 
$\Psi(0) = 0$ if and only if 
$\sigma =0$. Finally, there exists $\th_0=\th_0(\delta)$
such that $\lim_{m\to\infty}\Psi'(m)<1$ if and only if $\th>\th_0$.
\end{lemma}

By specializing Theorem \ref{thm:Risk} to the case 
$\psi(x_1,x_2) = (x_1-x_2)^2$ and using the fixed point condition 
$m_* = \Psi(m_*;\delta,\sigma,\nu,\th_*)$ we obtain immediately
the following.
\begin{corollary}\label{coro:MSE}
Let $ \Seq \in CS(\delta,\sigma,\nu,\omega)$ be a convergent sequence of problem
instances, and further assume that $\ayyenn \sim \Gauss(n,N)$. 
Denote by $\hxlenn$ the \LASSO\, estimator for problem
instance $I_{n,N}$, with $\lambda\ge 0$ .
Then, almost surely
\begin{eqnarray}\label{eq:asymptotic-result-MSE}
\lim_{N\to\infty}\frac{1}{N}
\|\hx_{\lambda}-x_{0}\|_2^2 = m_*\, ,
\end{eqnarray}
where $m_* = \HFP(\Psi(\,\cdot\, ;\delta,\sigma,\nu,\th_*))$,
and $\th_*=\th_*(\lambda)$ is fixed by the calibration relation
(\ref{eq:Calibration}).
\end{corollary}
%
%***********************************************************
%
\subsection{AMSE over  General Sequences}
\label{sec:AMSEGeneral}
Corollary \ref{coro:MSE} determines the asymptotic mean square error for 
convergent sequences $\Seq \in CS(\delta,\sigma,\nu)$. 
The resulting expression depends on $\delta,\sigma,\nu$,
and is denoted  $\AMSE_{SE}(\lambda;\delta,\sigma,\nu)$. 
We have
\begin{eqnarray}
\AMSE_{SE}(\lambda;\delta,\sigma,\nu) =  \HFP(\Psi(\,\cdot\, ;\delta,\sigma,\nu,\th_*)).
\label{eq:StateEvolutionAMSE}
\end{eqnarray}
The introduction considered instead the 
asymptotic mean square error $\AMSE(\lambda;\Seq)$
along general, not necessarily convergent sequences of problem instances
in the standard $\ell_p$ problem suite $\Seq \in\cS_p(\delta,\xi,\sigma)$, cf. 
Eq.~(\ref{eq:AMSE_Def}). %,
%which we recall here for the reader's convenience. 
%With $\bx$ again denoting a sequence of objects $x^{(N)} \in \bR^N$ set
%%
%\begin{eqnarray}
%%
%\cX_p(\xi) &\equiv & \cup_{B=1}^{\infty}\cX_p(\xi;B)\, ,\\
%\cX_p(\xi;B)& \equiv &
%\Big\{\bx\, :\, \|\exxohenn\|_p^p\le N\xi^p
%\mbox{ for all } N \; ;\;
%\|x_{0}^{(N)}\|_2^2\le BN\Big\}\, .
%%
%\end{eqnarray}
%%
%In words, the uniform integrability condition amounts to the existence
%of a uniform bound (of any size) on the mean square 
%value of entries $x_{0,i}^{(N)}$.
Given a sequence $\Seq\in\cS_p(\delta,\xi,\sigma)$, we let
\begin{eqnarray}
\AMSE(\lambda;\Seq) = \lim\sup_{N\to\infty}
\frac{1}{N}\E\big\{\|\hx_{\lambda}^{(N)}-\exxohenn\|^2\big\} 
\, .\label{eq:SequenceAMSE}
\end{eqnarray}
Below we will often omit  the subscript ${\sc SE}$ on $AMSE_{SE}$,
thereby using the same notation for the state evolution quantity
(\ref{eq:StateEvolutionAMSE}) and the sequence quantity 
(\ref{eq:SequenceAMSE}).
This abuse is justified by the following key fact. The asymptotic
mean square error along \emph{any} sequence of instances can be represented
by the formula $\AMSE_{SE}(\lambda;\delta,\nu,\sigma)$, for a suitable $\nu$ -- provided the sensing 
matrices $\ayyenn$ have i.i.d. Gaussian entries.
Before stating this result formally, we 
recall that the definition of sparsity class $\cF_p(\xi)$ was
given in Eq.~(\ref{eq:SparseProbs}).
\begin{proposition}\label{propo:Converging}
Let $\Seq$ be any {\em (not necessarily convergent)} sequence of problem
instances in $\cS_p(\delta,\xi,\sigma)$. Then there exists 
a probability distribution $\nu\in\cF_p(\xi)$ such that
\begin{eqnarray}
\AMSE(\lambda;\Seq) = \AMSE_{SE}(\lambda;\delta,\nu,\sigma)  ,
\end{eqnarray}
and both sides are given by 
the fixed point of the one-dimensional map $\Psi$, 
namely\\ $\HFP(\Psi(\,\cdot\, ;\delta,\sigma,\nu,\th_*))$.
Further, for each $\eps > 0$,
\begin{eqnarray}
 \lim\sup_{N\to\infty}\prob\Big\{
\frac{1}{N}\|\hx_{\lambda}^{(N)}-\exxohenn\|_2^2\ge \AMSE_{SE}(\lambda;\delta,\nu,\sigma)+\eps   
\Big\} =0\, .\label{eq:InProb}
\end{eqnarray}

Conversely, for any $\nu\in \cF_p(\xi)$, there exists a sequence
of instances $\Seq\in\cS_p(\delta,\xi,\sigma)$, such that 
$\AMSE(\lambda;\Seq) = \AMSE(\lambda;\delta,\nu,\sigma)$ 
along that sequence.
\end{proposition}
\begin{proof}
Given the sequence of problem instances, $\Seq= \{ \exxohenn, \zeeenn, \ayyenn\}_{n,N}$,
extract a subsequence along which 
the expected mean square error has a limit equal to the $\lim\sup$ in
Eq.~(\ref{eq:AMSE_Def}). We will then extract a further  subsequence 
 that is a 
convergent  subsequence of problem instances, in the sense of Definition 
\ref{def:Converging}, hence proving the direct part of
our claim, by virtue of Corollary \ref{coro:MSE}.
(Convergence of the expectation of $\|\hx_{\lambda}^{(N)}-\exxohenn\|^2/N$
follows from almost sure convergence together with the fact that 
$\|\exxohenn\|^2/N$ is uniformly bounded by assumption and
$\|\hx_{\lambda}^{(N)}\|^2/N$ is uniformly bounded by
Lemma 3.3 in \cite{BayatiMontanariLASSO}.)

Let $\nu_{x_0,N}$ be the empirical distribution of $\exxohenn$ as in (\ref{ecdf-def}).
%
%\begin{eqnarray}
%
%\nu_{x_0,N} \equiv \frac{1}{N}\sum_{i=1}^N\delta_{x_{0,i}^{(N)}}\, .
%
%\end{eqnarray}
%
Since $\Seq\in\cS_p(\delta,\xi,\sigma)$, we have $\nu_{x_0,N}(|X|^p)\le \xi^p$
hence the family $\{\nu_{x_0,N}\}$ is tight, and along a further subsequence 
the empirical distributions of $\exxohenn$ converge weakly,  to a limit $\nu$, say. 
Again by $\Seq\in\cS_p(\delta,\xi,\sigma)$, the empirical distributions of $\zeeenn$  are tight  (assumption 
$\bz \in \cZ^2(\sigma)$ entails $\|\zeeenn\|^2/n\to\sigma^2$); we extract yet another 
subsequence along which they converge, to $\omega$, say. 

We are left with a subsequence we shall label
$\{(n_k,N_k)\}_{k\ge 1}$.  We wish to prove for this sequence
 (a)-(c) of  Definition \ref{def:Converging}.
Property (c) in  Definition \ref{def:Converging}, the convergence of column norms, is well known to hold
for random matrices with iid Gaussian entries (and easy to show).
We are left to show (a) and (b), i.e. that $\nu_{x_0,N_k}(X^2)\to \nu(X^2)$
 and $\nu_{z,n_k}(X^2)\to \omega(X^2)$ along this sequence.
Convergence of the second moments follows since
\begin{eqnarray*}
\lim_{k\to\infty}\nu_{x_0,N_k}(X^2) &= \lim_{k\to\infty}\nu_{x_0,N_k}(X^2\ind_{\{|X|\le M\}})
+\lim_{k\to\infty}\nu_{x_0,N_k}(X^2\ind_{\{|X|> M\}})\\
&= \nu(X^2\ind_{\{|X|\le M\}})+{\sf err}_M
\end{eqnarray*}
where we used the dominated convergence theore, 
where, by the uniform integrability 
property of sequences  $\bx$ in $\cX_p(\xi)$, ${\sf err}_M\le \epsilon_M
\downarrow 0$ as $M\to\infty$. 

The limit in probability \eqref{eq:InProb} follows by very similar arguments
and we omit it here.

The converse is proved by taking $\exxohenn$ to be a vector with iid
components $\exxohenn\sim \nu$. The empirical distributions $\nu_N$
then converge almost surely to $\nu$ by the Glivenko-Cantelli theorem.
Convergence of second moments follows from the strong
law of large numbers.
\end{proof}

%
%***********************************************************
%
\subsection{Intuition and relation to AMP algorithm}

Theorem \ref{thm:Risk} implies that, in the high-dimensional limit,
vector estimation through the LASSO can be effectively understood 
in terms of $N$ uncoupled scalar estimation problems, provided the 
noise is augmented by an undersampling-dependent increment.
A natural question is whether one can construct, starting 
from the vector of measurements $y = (y_1,\dots,y_n)$
(which are intrinsicaly `joint' measurements of $x_1,\dots,x_N$), 
a collection of $N$ uncoupled measurements of $x_1,\dots,x_N$.

A deeper intuition about this question and Theorem \ref{thm:Risk} can be developed 
by considering the approximate message passing (AMP) algorithm
first introduced in \cite{DMM09}. At one given problem instance (i.e. frozen choice of $(n,N)$)
we omit the superscript $(N)$.  The algorithm produces a sequence of 
estimates $\{\hx^0, \hx^1, \hx^2\,\dots\}$ in $\reals^N$,
by letting $\hx^0=0$ and, for each $t\ge 0$
\begin{eqnarray}
z^t & = & y -A\hx^t + \frac{\|\hx^t\|_0}{n}\, z^{t-1} 
\label{eq:FOAMP2} \\
\hx^{t+1} & = & \eta(\hx^t+A^Tz^t;\theta_t)\, ,
\label{eq:FOAMP1}
\end{eqnarray}
where $\|\hx^t\|_0$ is the size of the support of $\hx^t$.
Here $\{z^t\}_{t\ge 0}\subseteq \reals^n$ is a sequence of residuals 
and $\theta_t$ a sequence of thresholds. 

As shown in \cite{BM-MPCS-2010}, the vector
$\hx^t+A^Tz^t$  is distributed asymptotically  (large $t$) as 
$x_0+w^t$ with $w^t\in\reals^N$ a vector with i.i.d. components
$w_i^t\sim \normal(0,\sigma_t^2)$ independent of $x_0$.
(Here the convergence is to be understood in the sense of finite-dimensional
marginals.)  In other words, 
\emph{the vector $\hx^t+A^Tz^t$ produced
by the AMP algorithm is effectively a vector of
i.i.d. uncoupled observations of the signal $x_0$.}

The second key point is that the AMP algorithm is tightly 
related to the LASSO. First of all, fixed points of AMP 
(for a fixed value of the threshold $\theta_t=\theta_{*}$) are
minimizers of the LASSO cost function and viceversa,
provided the $\theta_{*}$ is calibrated with the regularization parameter 
$\lambda$ according to the following relation
\begin{eqnarray}
\lambda = \theta_* \cdot \Big(1-\frac{\|\hx_{\lambda}\|_0}{n}\Big),
\end{eqnarray}
with $\hx_{\lambda}$ the LASSO minimizer or --equivalently--
the  AMP fixed point. 
Finally, \cite{BayatiMontanariLASSO} proved that (for Gaussian 
sensing matrices $A$), the AMP estimates do converge to the 
LASSO minimizer provided the sequence of thresholds is chosen according 
to the policy
\begin{equation} \label{ThreshChoice}
   \theta_t = \th \, \sigma_t \, ,
\end{equation}
for a suitable $\alpha>0$ depending on $\lambda$ 
\cite{BayatiMontanariLASSO,NSPT}. Finally,  the effective 
noise-plus-interference level
$\sigma_t$ can be estimated in several ways, a simple one being 
$\widehat{\sigma}_t^2= \|z_t\|^2/n$.
%
%************************************************************
%

\section{Minimax MSE over $\ellpee$ Balls, Noiseless Case}
\label{sec:MM_Noiseless}

In this section we state results for the noiseless case, $y = Ax_0$,
where $A$ is $n \times N$ and $x_0$ obeys an $\ellpee$ constraint.
As mentioned in the introduction, our
results hold in the asymptotic regime where 
$n/N \goto \delta \in (0,1)$.
%
%*************************
%
\subsection{Main Result}

Let $\Seq \equiv \{(\exxohenn, \zeeenn, \ayyenn)\}_{n,N}$
be a sequence of {\it noiseless} problem instances ($\zeeenn=0$: no noise is added to the 
measurements) with Gaussian sensing matrices $\ayyenn \sim \Gauss(n,N)$.
Define the minimax LASSO mean square error as 
\begin{eqnarray}
M_p^*(\delta,\xi) \equiv \sup_{\Seq \in \cS_p(\delta,\xi,0)}  \inf_{\lambda\in\reals_+} 
\AMSE(\lambda;\Seq)\, .
\end{eqnarray}

\begin{theorem} \label{thm:ellp:noiseless}
Fix $\delta \in(0,1)$, $\xi > 0$. The minimax AMSE obeys:
\beq \label{eq:ellp:nonoise:minmax:eval}
    M_p^*(\delta,\xi) = \frac{\delta \xi^2}{M_p^{-1}(\delta)^2}.
\eeq
Further we have:

\noindent{\bf Minimax Threshold.} The minimax threshold $\lambda^*(\delta,\xi)$
is given by the calibration relation (\ref{eq:Calibration}) 
with $\th = \th^*(\delta,\xi)$ determined as follows
(notice in particular that this is independent of $\xi$):
\begin{eqnarray}
     \th^*(\delta,\xi) = \th_p(M_p^{-1}(\delta)) \, .
\label{eq:MinimaxThresholdNoiseless}
\end{eqnarray}

\noindent{\bf Least Favorable $\nu$.}  The least-favorable distribution is 
a 3-point distribution $\nu^* = \nu_{p,\delta,\xi}^* 
= \nu_{\eps^*,\mu^*}$ (cf. Eq.~(\ref{eq:3-point})) with
\begin{eqnarray}
\mu^*(\delta,\xi) =  \frac{\xi}{M_p^{-1}(\delta)}\,  \mu_p(M^{-1}_p(\delta)) \, ,\;\;\;\;\;
\eps^*(\delta,\xi) = \frac{\xi^p}{(\mu^*)^p}\, .
\end{eqnarray}

\noindent{\bf Saddlepoint.}  The above quantities obey a saddlepoint relation.
Put for short $\AMSE(\lambda;\nu)$ in place of $\AMSE(\lambda;\delta,\nu,0)$,
The minimax AMSE obeys
\[
M_p(\delta;\xi) = \AMSE(\lambda^*;  \nu^*)
\]
and
\begin{eqnarray}
 \AMSE(\lambda^*;  \nu^*) &\leq&  \AMSE(\lambda;  \nu^*)\, ,
\qquad \forall \lambda > 0 \\
              & \geq&   \AMSE(\lambda^*; \nu)
\, , \qquad \forall \nu \in \effpee(\xi) . 
\end{eqnarray}
\end{theorem}
%
%************************************************************
%
\subsection {Interpretation}
\label{sec:AsymptoticNoiseless}

Figure \ref{fig:MMxMSE} presents the function $M_p^*(\delta,\xi =1)$ 
on a logarithmic scale.
As the reader can see, there is a substantial increase in 
the minimax risk as $\delta \goto 0$,
which agrees with our intuitive picture that the reconstruction becomes 
less accurate for small $\delta$ (high undersampling).

\begin{figure}
\begin{center}
\includegraphics[height=3.0in]{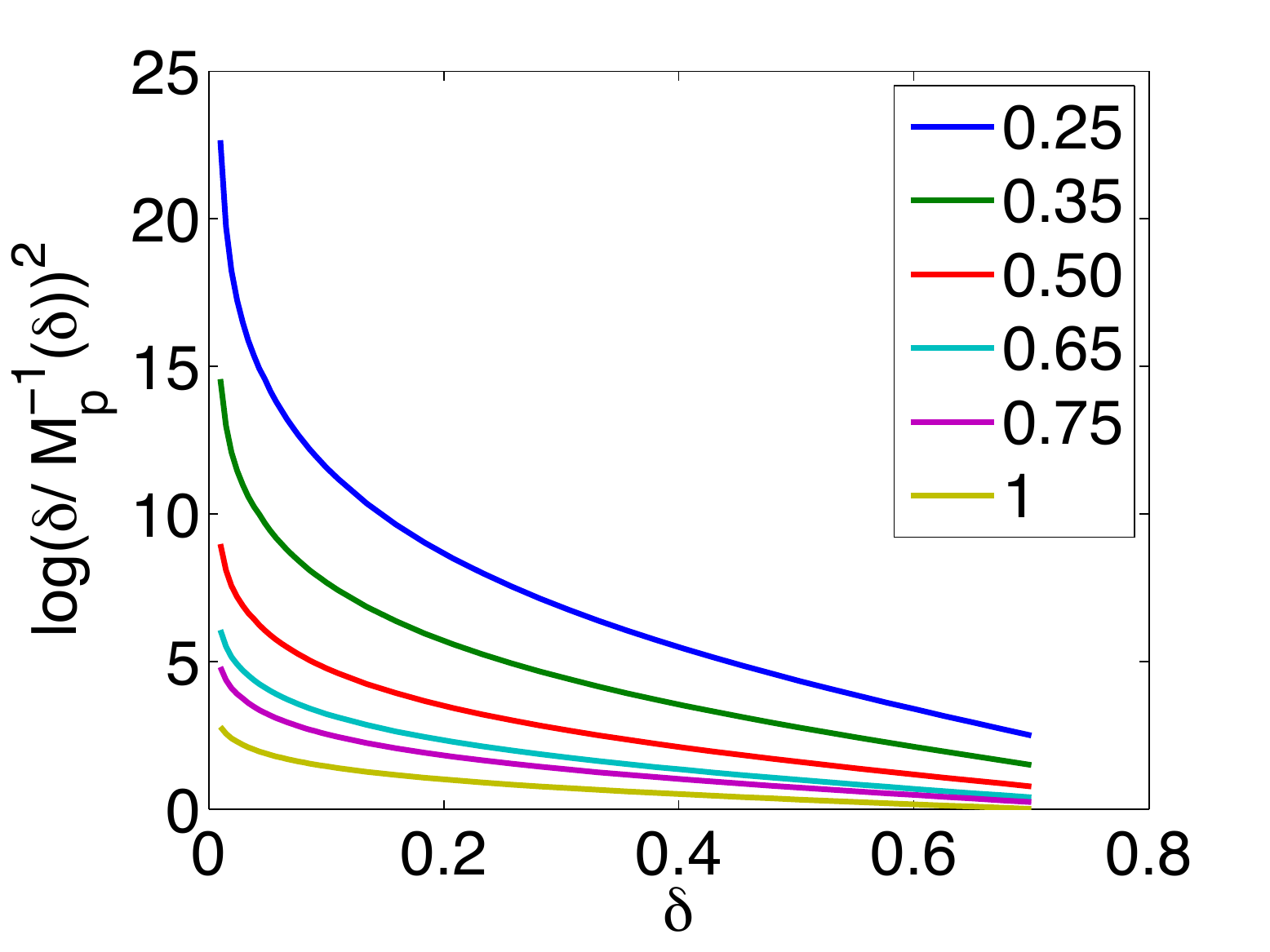} 
\caption{Minimax MSE $M_p^*(\delta,1)$. We assume here $\xi=1$; curves show
$\log$ MSE as a function of $\delta$. Consistent with $\delta \goto 0$
asymptotic theory, the curves are nearly scaled copies of each other.}
\label{fig:MMxMSE}
\end{center}
\end{figure}

The asymptotic properties of $M_p^*(\delta,1)$ 
in the high undersampling regime ($\delta\to 0$) can be derived 
using Lemma  \ref{lem:asymp:scalr:minimax}. From 
Eq.~(\ref{eq:asymp:MpInv}) we have
\begin{eqnarray*}
      M_p^*(\delta,1) = 
\delta^{1-2/p} (2\log(\delta^{-1}))^{2/p-1}\big\{1+o_{\delta}(1)\big\}, 
\qquad \delta \goto 0.
\end{eqnarray*}
Hence, when plotting $\log  M_p^*(\delta,1)$, as we do here, we should see
graphs of the form 
\[
  \log M_p^*(\delta,1) =   (1-2/p) \cdot  \left[ \log(\delta^{-1}) - \log( \log(\delta^{-1}))  - \log(2)\right]+o_{\delta}(1)  , \qquad \delta \goto 0.
\]
In particular the curves should look `all the same' at small $\delta$,
 except for
scaling; this is qualitatively consistent with Fig \ref{fig:MMxMSE}, even at 
larger $\delta$.

Another useful prediction can be obtained by working out the 
asymptotics of the minimax threshold  $\lambda^*(\delta,\xi)$.
Using Eq.~(\ref{eq:MinimaxThresholdNoiseless}) as well as the 
calibration relation (\ref{eq:Calibration}), we get,
as $\delta\to 0$,
\begin{eqnarray}
\lambda^*(\delta,\xi) = \xi\,\cdot\Big(
\frac{2\log (1/\delta)}{\delta}\Big)^{1/p}\,\big\{1+o_{\delta}(1)\big\}\, .    \label{eq:asymp:lambda}
\end{eqnarray}
%

%
%
%********************************************************
%
\subsection{Proof of Theorem \ref{thm:ellp:noiseless}}

We will focus on proving Eq.~(\ref{eq:ellp:nonoise:minmax:eval}),
since the other points follow straightforwardly.
By Proposition \ref{propo:Converging}, we have the equivalent 
characterization 
\begin{eqnarray}
M_p^*(\delta,\xi) 
= \sup_{\nu \in \cF_p(\xi)}  \inf_{\lambda\in\reals_+} 
\AMSE(\lambda;\delta,\nu,\sigma=0)\, .\label{eq:MinmaxEquivalent}
\end{eqnarray}
Further, by Corollary \ref{coro:MSE}, we can use the mean square error
expression given there, and because of the monotone nature of the
calibration relation, we can minimize over the threshold $\tau$ instead of
$\lambda$. We get therefore
\begin{eqnarray}
M_p^*(\delta,\xi) 
= \sup_{\nu \in \cF_p(\xi)}  \inf_{\th\in\reals_+} 
\AMSE_{\SE}(\th;\delta,\nu,0)\, ,
\end{eqnarray}
where 
\begin{align}
\AMSE_{\SE}(\th;\delta,\nu,0)&  = m\, ,\nonumber\\
m & = \stMSE\big(m/\delta;\nu,\th)\, .\label{eq:FixedPointProof}
\end{align}

\begin{figure}
\includegraphics{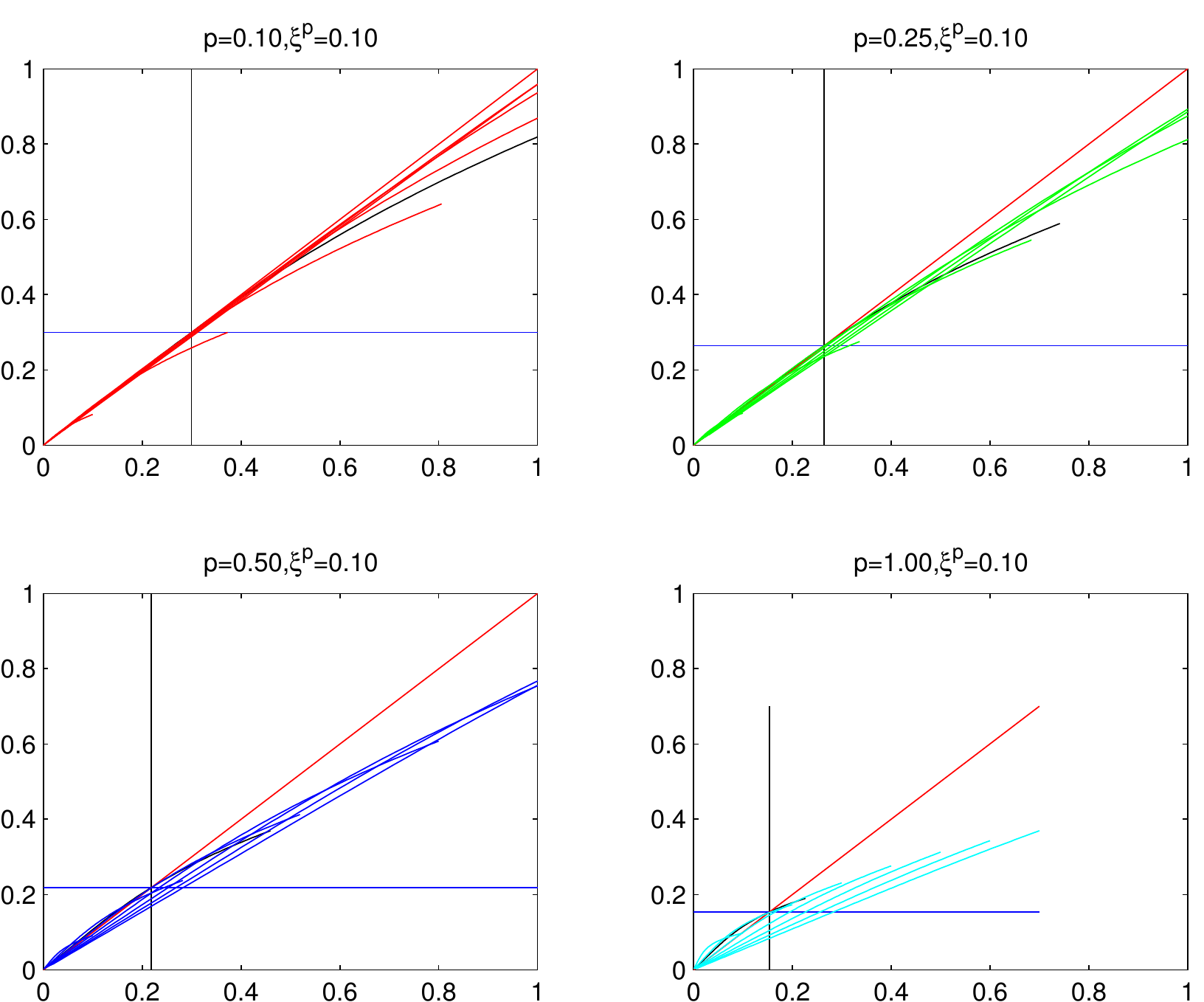}
\caption{Illustration of the minimax fixed point property.
Horizontal input MSE $m$.
Vertical: output MSE $\Psi(m)$ for the 
state evolution map defined as per Eq.~(\ref{eq:PsiDef}).
Red diagonal: $\Psi(m) = m$.
Black vertical Line: minimax HFP $M_p(\xi)$.
Blue horizontal  line: minimax MSE $M_p(m)$.
Black curve: MSE map at minimax threshold value and least-favorable
distribution. It crosses
the diagonal at the minimax fixed point.
Colored Curves. MSE maps at minimax threshold value and
other three-point distributions. All other fixed points occur below $M_p(\xi)$.}
\end{figure}

\begin{figure}
\includegraphics{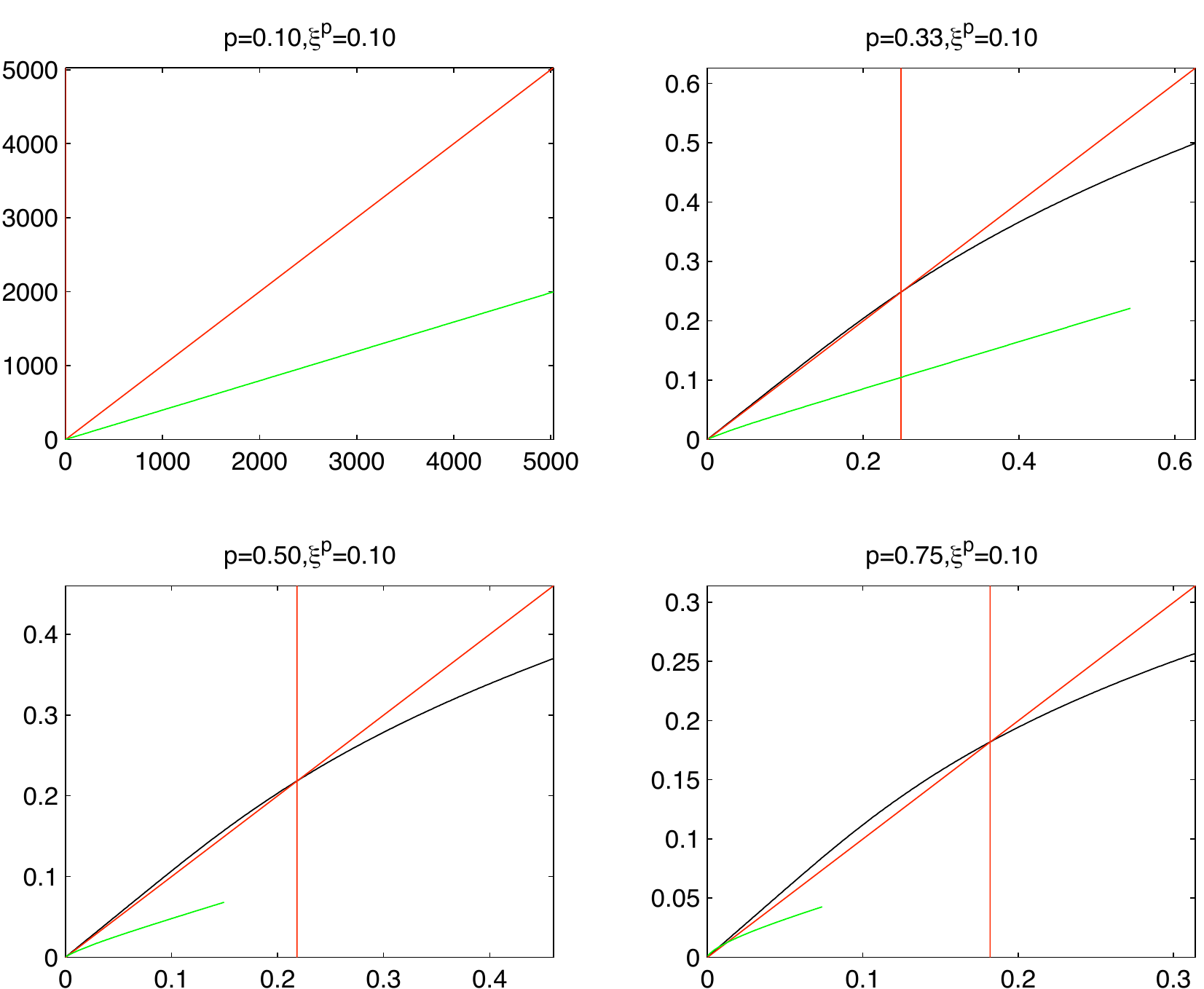}
\caption{Comparisons of highest fixed point at power law distribution
with minimax HFP. 
Horizontal input MSE $m$.
Vertical: output MSE $\Psi(m)$ for the 
state evolution map defined as per Eq.~(\ref{eq:PsiDef}).
Red diagonal: $\Psi(m) = m$.
Red vertical: Minimax MSE. 
Black curve: MSE map at minimax threshold value and least-favorable
distribution. It crosses
the diagonal at the minimax fixed point.
Green Curve: MSE map with same threshold,
taken at power law distribution calibrated to same $\E |X|^p = \xi^p$
constraint.} 
\end{figure}

Recall that $\cP(\reals)$ denotes the class of all probability distribution 
functions on $\reals$.
Define the scaling operator  $S_a : \cP(\reals) \to \cP(\reals)$ by 
$(S_a\nu)(B) = \nu(B/a)$ for any Borel set $B$.
For the family of operators $\{ S_a : a > 0 \}$
 we have the group properties
\beq
 \label{group}
 S_a \cdot S_b = S_{a \cdot b } , \qquad S_1 = \id , \qquad S_{a} S_{a^{-1}} = S_1 .
\eeq
In particular by the last property, for any  $a > 0$, the operator 
$S_a: \cP(\reals) \mapsto \cP(\reals)$ is one-to-one.

With this notation, we  have the scale covariance property
of the soft-thresholding mean square error
\begin{eqnarray}
\stMSE(\sigma^2 ; \nu, \th) = \sigma^2 \cdot \stMSE(1; S_{1/\sigma} \nu, \th),
\end{eqnarray}
transforming a general-noise-level problem into a noise-level-one problem.
As a consequence of Lemma \ref{lemma:Map}, the map 
$\sigma^2\mapsto\stMSE(\sigma^2 ; \nu, \th)$ is (for fixed
$\nu, \th$) increasing and concave. Therefore, the map
$\sigma^2\mapsto \stMSE(1; S_{1/\sigma} \nu, \th)$ is 
strictly monotone decreasing.
Also, the fixed point Eq.~(\ref{eq:FixedPointProof}) can be rewritten as
\begin{eqnarray}
\delta = \stMSE(1;S_{\sqrt{\delta/m}}\nu,\tau)\, ,\label{eq:FixedPointRewritten}
\end{eqnarray}
where the solution is unique by strict monotonicity of 
$m\mapsto \stMSE(1;S_{\sqrt{\delta/m}}\nu,\tau)$.

We will prove Eq.~(\ref{eq:ellp:nonoise:minmax:eval}) by
obtaining an upper and a lower bound for $M_p(\delta,\xi)$.
In the following we assume without loss of generality
that the infimum in Eq.~(\ref{eq:MinmaxEquivalent}) is achieved
\begin{eqnarray}
M_p(\delta,\xi) = \AMSE_{\SE}(\th_*;\nu_*,0)\, .\label{eq:MMAXPoint}
\end{eqnarray}
Further we will use the minimax conditions for soft thresholding,
see Lemma \ref{lemma:ScalarSaddle}:
\begin{align}
\inf_{\tau\in\reals_+}\stMSE(1;\nu,\th)&\le M_p(\xi) \, ,\;\;\;\;\;
\forall\, \nu\in\cF_p(\xi)\, ,\label{eq:InfTau}\\
\sup_{\nu\in\cF_p(\xi)}\stMSE(1;\nu,\th)&\ge M_p(\xi) \, ,\;\;\;\;\;
\forall\, \tau\in\reals_+\, .\label{eq:SupNu}
\end{align}

\noindent{\bf Upper bound on $M_p(\delta,\xi)$.} Let $m_* =M_p(\delta,\xi) = 
\AMSE_{\SE}(\th_*;\nu_*,0)$. By Eq.~(\ref{eq:FixedPointProof})
and (\ref{eq:FixedPointRewritten}) we have
\begin{eqnarray}
\delta = \stMSE(1;S_{\sqrt{\delta/m_*}}\nu_*,\tau_*) =
\inf_{\tau\in\reals_+} \stMSE(1;S_{\sqrt{\delta/m_*}}\nu_*,\tau)\, . 
\end{eqnarray}
The second equality follows because otherwise by there
would exist $\tau_{**}$ with $\stMSE(1;S_{\sqrt{\delta/m_*}}\nu_*,\tau_{**})$
whence, by the monotonicity of $m\mapsto \stMSE(1;S_{\sqrt{\delta/m}}\nu_*,\tau_{**})$ it would follow that $\AMSE_{\SE}(\th_{**};\nu_*,0)<
\AMSE_{\SE}(\th_*;\nu_*,0)$ which violates the minimax property
(\ref{eq:MMAXPoint}). 

Next notice that $S_{\sqrt{\delta/m_*}}\nu_*\in\cF_p(\sqrt{\delta/m_*}\,\xi)$
whence by Eq.~(\ref{eq:InfTau}), we get $\delta\le M_p(\sqrt{\delta/m_*}\,\xi)$.
By the monotonicity of $\xi\mapsto M_p(\xi)$ this yields
\begin{eqnarray}
m_*\le \frac{\delta\,\xi^2}{M_p^{-1}(\delta)^2}\, .
\end{eqnarray}

\noindent{\bf Lower bound on $M_p(\delta,\xi)$.} Again by Eq.~(\ref{eq:FixedPointProof})
and (\ref{eq:FixedPointRewritten}) we have
\begin{eqnarray}
\delta = \stMSE(1;S_{\sqrt{\delta/m_*}}\nu_*,\tau_*) =
\sup_{\nu\in\cF_p(\xi)} \stMSE(1;S_{\sqrt{\delta/m_*}}\nu,\tau_*(\nu))\, ,
\end{eqnarray}
with $\tau_*(\nu)$ the optimal threshold
for distribution $\nu$ and 
the second equality following by an argument similar to the 
one above (i.e. if this weren't true, there would be a {\it different} worst distribution 
$\nu_{**}$, reaching contradiction). But $\nu\in\cF_p(\xi)$ implies 
$S_{\sqrt{\delta/m_*}}\nu\in\cF_p(\sqrt{\delta/m_*}\,\xi)$, whence
\begin{eqnarray}
\delta =
\sup_{\nu\in\cF_p(\xi\sqrt{\delta/m_*})} \stMSE(1;\nu,\tau_*(\nu))
\ge M_p(\sqrt{\delta/m_*}\,\xi) \, ,
\end{eqnarray}
where the second inequality follows by Eq.~(\ref{eq:SupNu}).
The proof is finished by using again
the monotonicity of $\xi\mapsto M_p(\xi)$.
\qed
%
%******************************************************************
%
\section{Minimax MSE over $\ellpee$ Balls, Noisy Case}
\label{sec:MM_Noisy}

In this section we generalize the results of the previous 
section to the case of noisy measurements with noise variance 
per coordinate equal to $\sigma^2$.

\subsection{Main Result}

Now let $\sigma > 0$ and consider  
sequences  $\Seq$ of  {\it noisy} problem instances from the standard $\ell_p$ problem
suite $\Seq \in \cS_p(\delta,\xi,\sigma)$; hence, in addition to the $\ell_p$ constraint
$\|x_0^{(N)}\|_p^p  \leq  N \xi^p$ and each $\ayyenn \sim \Gauss(n,N)$,
 now the noise vectors $\zeeenn\in\reals^n$ are non-vanishing
and have norms satisfying $\|\zeeenn\|^2/n\to \sigma^2 > 0$.

We define the minimax LASSO asymptotic mean square error
as 
\begin{eqnarray}
M_p^*(\delta,\xi,\sigma) \equiv \sup_{\Seq \in \cS_p(\delta,\xi,\sigma)}  \inf_{\lambda\in\reals_+} 
\AMSE(\lambda;\Seq)\, .
\end{eqnarray}
By simple scaling of the problem we have, for any $\sigma>0$,
\begin{eqnarray}
M_p^*(\delta,\xi,\sigma) = \sigma^2\, M_p^*(\delta,\xi/\sigma,1)\, ,
\label{eq:NoiseOne}
\end{eqnarray}
an observation which will be used repeatedly in the following.

\begin{theorem} \label{thm:ellp:noisy}
For any $\delta,\xi>0$, let $m^* = m^*_p(\delta,\xi)$ be the unique positive 
solution of
\begin{eqnarray}
\frac{m^*}{1+m^*/\delta} = M_p\left(\frac{\xi}{(1+m^*/\delta)^{1/2}}\right)\, .
\label{eq:MainthmEq}
\end{eqnarray}
Then the LASSO minimax mean square error $M_p^*$is given by:
\beq 
\label{eq:ellp:minmax:eval}
M_p^*(\delta, \xi,\sigma) =  \sigma^2 \cdot m^*(\delta,\xi/\sigma)\, . 
\eeq
Further, denoting by $\xi^*\equiv(1+m^*/\delta)^{-1/2}\xi/\sigma$,
  we have:

\noindent{\bf Least Favorable $\nu$.}  The least-favorable distribution  
is a  $3$-point mixture $\nu^* = \nu^*_{p,\delta,\xi,\sigma}=
\nu_{\eps^*,\mu^*}$ (cf. Eq.~(\ref{eq:3-point})) with
\begin{eqnarray}
\mu^*(\delta,\xi,\sigma) = \sigma \cdot (1+m^*/\delta)^{1/2} \mu_p(\xi^*) ,\;\;\;\;\;
\eps^*(\delta,\xi,\sigma) = \frac{\xi^p}{(\mu^*)^p} \, ,
\end{eqnarray}
with $m^* = m^*(\delta,\xi/\sigma)$ given by the solution of
Eq.~(\ref{eq:MainthmEq})

\noindent{\bf Minimax Threshold.} The minimax threshold 
$\lambda^*(\delta,\xi,\sigma)$ is given by the calibration relation 
(\ref{eq:Calibration}) 
with $\th = \th^*(\delta,\xi,\sigma)$ determined as follows:
\begin{align}
\th^*(\delta,\xi,\sigma) = \th_p(\xi^*) .
\end{align}
with  $\th_p(\,\cdot\,)$ the
soft thresholding minimax threshold, 
$\xi^*\equiv(1+m^*/\delta)^{-1/2}\xi/\sigma$ and $\nu=\nu^*$
is the least favorable distribution given above.

\noindent {\bf Saddlepoint.}  The above quantities obey a saddlepoint relation.
Put for short $\AMSE(\lambda;\nu)$ in place of $\AMSE(\lambda;\delta,\nu,\sigma)$.
The minimax AMSE obeys
\[
M_p(\delta,\xi,\sigma) = \AMSE(\lambda^*; \nu^*)\, ,
\]
and
\begin{eqnarray}
 \AMSE(\lambda^*; \nu^*) &\leq&  \AMSE(\lambda; \nu^*)\, ,
\qquad \forall \lambda > 0 \\
              & \geq&   \AMSE(\lambda^*; \nu)
\, , \qquad \forall \nu \in \effpee(\xi) . 
\end{eqnarray}
\end{theorem}
%
%*****************************************************
%
\subsection{Interpretation}

Figure \ref{fig:MMxMSENoisy} provides a concrete illustration of Theorem 
\ref{thm:ellp:noisy}. For various sparsity levels $\xi$
and undersampling factors $\delta$, the mean square
error $M_p(\delta,\xi,\sigma)$ can be easily computed.
As expected, the result is monotone increasing in $\xi$
and decreasing in $\delta$. For a given target mean square
error, such plots allow to determine the required number of 
linear measurements.

\begin{figure}
\begin{center}
\includegraphics[height=3.0in]{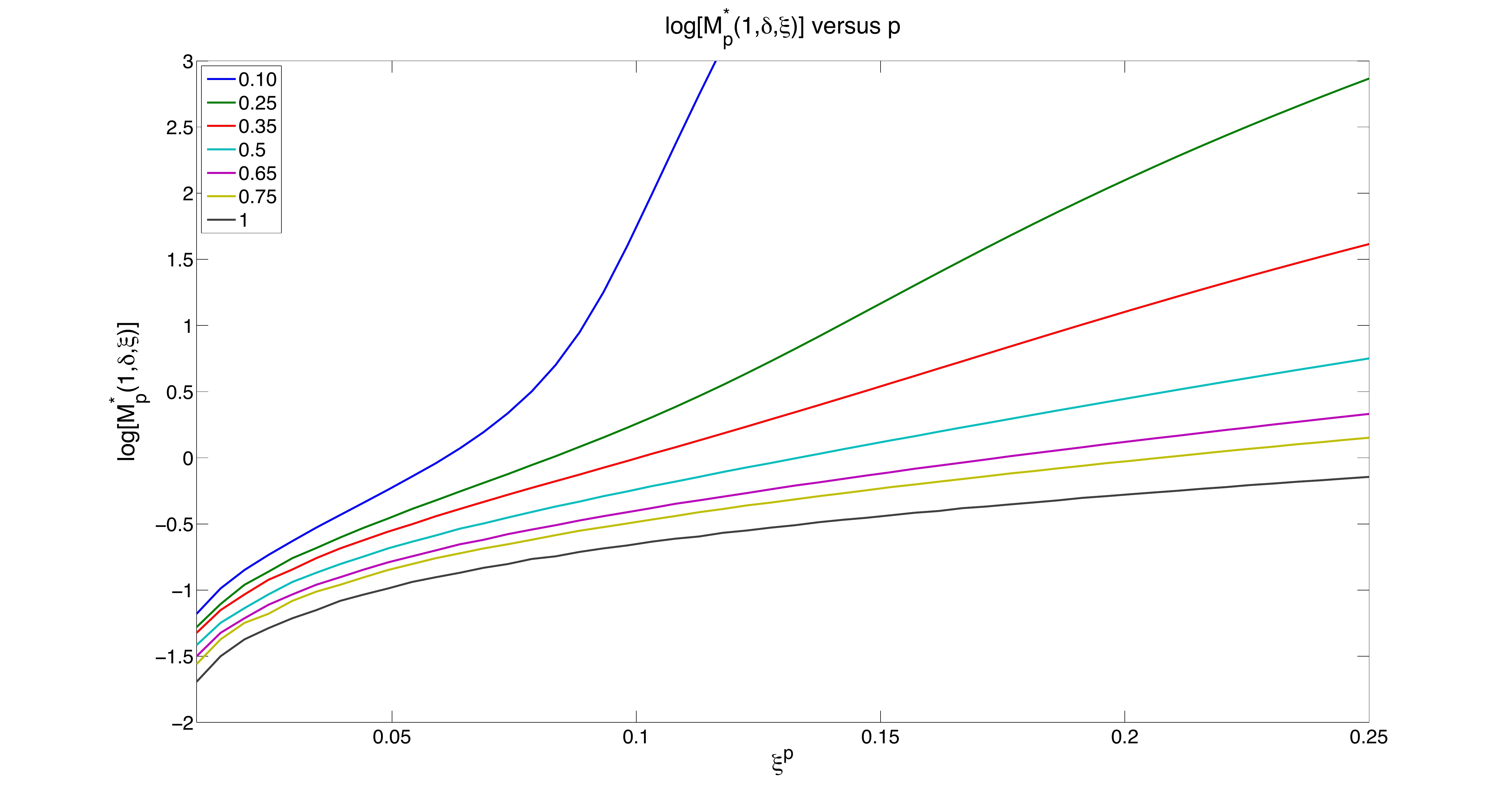} 
\caption{Minimax MSE $M_p^*(\delta,\xi,1)$, noisy case $\sigma=1$. 
We assume here $\delta=1/4$; curves show
$\MSE$ as a function of $\xi$. }
\label{fig:MMxMSENoisy}
\end{center}
\end{figure}

Equations (\ref{eq:MainthmEq}) and (\ref{eq:ellp:minmax:eval})
are somewhat more complex that their noiseless counterpart.
For this reason, it is instructive to work out the $\sigma\to 0$
limit $M_p^*(\delta,\xi,\sigma)$. By the basic scaling
relation (\ref{eq:ellp:minmax:eval}), this is equivalent to computing the
$\xi\to\infty$ limit of $M_p^*(1,\delta,\xi) = m^*(\delta,\xi)$.
Considering Eq.~(\ref{eq:MainthmEq}), it is easy to show that,
for large $\xi$
\begin{eqnarray*}
m^*(\delta,\xi) = c_0(\delta)\xi^2+c_1(\delta)+O(\xi^{-2})
\equiv c(\delta,\xi)\xi^2\, .
\end{eqnarray*}
Substituting in Eq.~(\ref{eq:MainthmEq})
\begin{eqnarray*}
\delta\left(1+\frac{\delta}{c\xi^2}\right)^{-1} = 
M_p\left((\delta/c)^{1/2}\left(1+\frac{\delta}{c\xi^2}\right)^{-1/2}\right)\, ,
\end{eqnarray*}
whence expanding for large $\xi$
\begin{eqnarray*}
\delta-\frac{\delta^2}{c_0\xi^2}+O(\xi^{-4}) = 
M_p((\delta/c_0)^{1/2})-\frac{\delta^{1/2}}{2\xi^2c_0^{3/2}}\, 
M_p'((\delta/c_0)^{1/2})\, (c_1+\delta) +O(\xi^{-4})\, .
\end{eqnarray*}
Imposing each order to vanish we get
\begin{eqnarray}
c_0(\delta) & = & \frac{\delta}{M_{p}^{-1}(\delta)}\, ,
\label{eq:FirstCoefficient}\\
c_1(\delta) & = &\frac{2\sqrt{c_0\delta}}{M_p'((\delta/c_0)^{1/2})}-\delta\, .
\label{eq:SecondCoefficient}
\end{eqnarray}
Our calculations can be summarized as follows.
\begin{corollary}
Fix a radius parameter $\xi$.
As $\sigma^2\to 0$, the asymptotic LASSO minimax mean square error
behaves as 
\begin{eqnarray}
M_p^*(\delta,\xi,\sigma) = \xi^2\, c_0(\delta) + \sigma^2\,
c_1(\delta) + O(\sigma^4/\xi^2)\, ,
\end{eqnarray}
with $c_0$ and $c_1$ determined by Eqs.~(\ref{eq:FirstCoefficient})
and (\ref{eq:SecondCoefficient}). In particular, in the high undersampling 
regime $\delta\to 0$, we get
\begin{eqnarray}
c_1(\delta) = \frac{2}{p}\{1+o_1(\delta)\}\, .\label{eq:C1delta}
\end{eqnarray}
\end{corollary}
The derivation of the asymptotic behavior (\ref{eq:C1delta}) is a 
straightforward calculus exercise, using Lemma
\ref{lem:asymp:scalr:minimax}. 

The last Corollary shows that the noiseless case,
cf. Theorem \ref{thm:ellp:noiseless} and 
Eq.~(\ref{eq:ellp:nonoise:minmax:eval}), is recovered as a special case
of the noisy case treated in this section. 
Further leading corrections due to small noise $\sigma^2\ll \xi^2$
are explicitly described by the coefficient $c_1(\delta)$
given in Eq.~(\ref{eq:SecondCoefficient}).

An alternative asymptotic of interest consists in
fixing the noise level $\sigma$, and letting $\xi/\sigma\to 0$.
In this regime the solution of Eq.~(\ref{eq:MainthmEq})
yields, using Lemma \ref{lem:asymp:scalr:minimax},
\begin{eqnarray}
m^*(\delta,\xi) = (2 \log( 1/\xi^{p}))^{1-p/2} \xi^p\cdot \{1+o(1)\}\, .
\end{eqnarray}
Substituting this expression in Theorem \ref{thm:ellp:noisy},
we obtain the following.
\begin{corollary}
Fix a noise parameter $\sigma^2>0$.
As $\xi\to 0$, the asymptotic LASSO minimax mean square error
behaves as 
\begin{eqnarray}
M_p^*(\delta,\xi,\sigma) = \sigma^{2-p}\xi^p \cdot 
\big\{2 \log\big(( \sigma/\xi)^p\big)\big\}^{1-p/2}\cdot \{1+o(1)\}\, ,
\end{eqnarray}
Further the minimax threshold value is given, in this limit,
by
\begin{eqnarray}
\lambda^* = \sigma \cdot \sqrt{2\log \big(( \sigma/\xi)^p\big)}\, 
\big\{1+o(1)
\big\}\, .
\end{eqnarray}
\end{corollary}
%
%****************************************************
%
\subsection{Proof of Theorem \ref{thm:ellp:noisy}}

The argument is structurally similar to the noiseless case.
We will focus again on proving the asymptotic expression
for minimax error given in Eq.~(\ref{eq:ellp:minmax:eval}),
since the other points of the theorem follow easily.
Using Proposition \ref{propo:Converging} and 
 Corollary \ref{coro:MSE},  the asymptotic mean square error
can be replaced by the expression given there and the minimization over
$\lambda$ can be replaced by a minimization over $\th$:
\begin{eqnarray}
M_p^*(\delta,\xi,\sigma) 
= \sup_{\nu \in \cF_p(\xi)}  \inf_{\th\in\reals_+} 
\AMSE_{\SE}(\th;\delta,\nu,\sigma)\, ,
\end{eqnarray}
where 
\begin{align}
\AMSE_{\SE}(\th;\delta,\nu,\sigma)&  = m\, ,\nonumber\\
m & = \stMSE\big(\sigma^2+m/\delta;\nu,\th)\, .\label{eq:FixedPointProof2}
\end{align}
By virtue of the scaling relation (\ref{eq:NoiseOne}),
we can focus on the case $\sigma^2=1$. Define, for all $m<\delta$
\begin{align}
\enn(m)   \equiv (1 + m/\delta)^{-1/2} \, .
\end{align}
We then have, applying Eq.~(\ref{eq:FixedPointProof2}) for the 
case $\sigma=1$, 
\begin{align}
\frac{m}{1+m/\delta} =  \stMSE\big(1;S_{\enn(m)}\nu,\th)\, .
\label{eq:FixedPointRescaled2}
\end{align}
Notice that $m\mapsto m/(1+m/\delta)$
is monotone increasing, and $m\mapsto\stMSE\big(1;S_{\enn(m)}\nu,\th)$
is monotone decreasing (because
$a^2\mapsto\stMSE\big(1;S_{1/a}\nu,\th)$ is decreasing as mentioned 
in the previous section). Hence this equation has a unique non-negative 
solution provided $\delta> \stMSE\big(1;\delta,\th)$, which
happens for all $\th>\th_0(\delta)$.

Assume without loss of generality that the minimax risk is achieved
by the pair $(\th_*,\nu_*)$. Then
\begin{eqnarray}
M_p^*(\delta,\xi,1) = \AMSE_{\SE}(\th_*;\delta,\nu_*,1) = m_*\, .
\label{eq:MMAXAssumption}
\end{eqnarray}
Then $m_*$ satisfies Eq.~(\ref{eq:FixedPointRescaled2}) with $\th = \th_*$
and $\nu = \nu_*$. 

\noindent{\bf Upper bound on $M_p(\delta,\xi,1)$.} By the last remarks,
we have
\begin{align}
\frac{m_*}{1+m_*/\delta} =  \stMSE\big(1;S_{\enn(m_*)}\nu_*,\th_*) = 
\inf_{\th \in\reals_+}  \stMSE\big(1;S_{\enn(m_*)}\nu_*,\th)\, .
\end{align}
The second equality follows from Eq.~(\ref{eq:MMAXAssumption}). Indeed 
if the equality did not hold, we could find $\th_{**}\in\reals_+$
such that  $\stMSE\big(1;S_{\enn(m_*)}\nu_*,\th_{**})<
 \stMSE\big(1;S_{\enn(m_*)}\nu_*,\th_*)$. But by the monotonicity 
of $m\mapsto m/(1+m/\delta)$ and of 
$m\mapsto  \stMSE\big(1;S_{\enn(m)}\nu_*,\th_{**})$, this would mean that the 
corresponding fixed point $m_{**}$ is strictly smaller than 
$m_*$. This would contradict the minimax assumption.

Since $S_{\enn(m_*)}\nu_*\in \cF_p(\enn(m_*)\xi)$ we can now apply
Eq.~(\ref{eq:InfTau}),  getting 
\begin{align}
\frac{m_*}{1+m_*/\delta} \le    M_p(\enn(m_*)\xi) \, .
\end{align}
Again by monotonicity of  $m\mapsto m/(1+m/\delta)$ and of $\xi\mapsto
M_p(\xi)$, this means that $m_*$ is upper bounded by the solution of 
Eq.~(\ref{eq:MainthmEq}).

\noindent{\bf Lower bound on $M_p(1,\delta,\xi)$.} 
Applying again Eq.~(\ref{eq:FixedPointRescaled2}) and an analogous argument
as above, we have
\begin{align}
\frac{m_*}{1+m_*/\delta} = 
\sup_{\nu\in\cF_p(\xi)}  \stMSE\big(1;S_{\enn(m_*)}\nu,\th_*(S_{\enn(m_*)}\nu))\, .
\end{align}
In the last expression $\th_*(S_{\enn(m_*)}\nu)$ is the optimal (minimal MSE)
threshold for distribution $S_{\enn(m_*)}\nu$.
For $\nu\in \cF_p(\xi)$, $S_{\enn(m_*)}\nu\in  \cF_p(\enn(m_*)\xi)$.
Further the map $S_{\enn(m_*)}:\cF_p(\xi)\to  \cF_p(\enn(m_*)\xi)$
is bijective. We thus have
\begin{align}
\frac{m_*}{1+m_*/\delta} = 
\sup_{\nu\in\cF_p(\enn(m_*)\xi)}  \stMSE\big(1;\nu,\th_*(\nu))\, .
\end{align}
By Eq.~(\ref{eq:SupNu}), we thus have 
\begin{align}
\frac{m_*}{1+m_*/\delta} \le    M_p(\enn(m_*)\xi) \, ,
\end{align}
which implies that $m_*$ is upper bounded by the solution of 
Eq.~(\ref{eq:MainthmEq}). This finishes our proof.
\qed
%
%******************************************************************
%
\section{Weak $p$-th Moment Constraints}

Our results for  $\ell_p$ constraints have natural 
counterparts for weak $\ell_p$ constraints.
We recall a standard definition for the weak-$\ell_p$ quasi-norm $\| x \|_{w\ell_p}$.
For a vector $x \in \reals^N$, let $T_x(t)  \equiv \{ i\in\{1,\dots,N\}\; :\; |x_i|  \geq  t \}$ 
index the entries of $x$ with amplitude above threshold $t$. Denoting by $|S|$ the 
cardinality of set $S$, we define
\begin{equation}
\| x \|_{w\ell_p} \equiv \max_{t\ge 0}\big[ \, t   |T_x(t)|^{1/p} \big]\, , \qquad 
\end{equation}
By Markov's inequality $ \| x \|_{w\ell_p} \leq  \| x \|_{p}$: 
the  weak $\ell_p$ quasi-norm
is indeed weaker than the $\ell_p$ norm (quasi norm, if $p < 1$).  
Weak-$\ell_p$ norms  arise frequently in 
applied harmonic analysis, as we discuss below.

As the reader no doubt expects, we can define a weak $\ell_p$ analogue to
the $\ell_p$ case.
\begin{definition}
\bitem
 \item {\bf Weak $\ell_p$ constraint.} 
A sequence $\bx_0 = (\exxohenn)$ belongs to  $\cX^w_p(\xi)$ if $(i)$ 
$\|x_0^{(N)}\|_{w\ell_p}^p \le N \xi^p$, for all $ N$;
and $(ii)$ there exists a sequence $B= \{B_M\}_{M\ge 0}$ such that
$B_M\to 0$, and for every $N$, $\sum_{i=1}^N(x_{0,i}^{(N)})^2\ind(|x_{0,i}^{(N)}|\ge M) \leq B_MN$.
\item {\bf Standard Weak-$\ell_p$ Problem Suite.}  Let $\cS_p^w(\delta,\xi,\sigma)$ 
denote the class of sequences of problem instances $I_{n,N} = ( \exxohenn,\zeeenn,\ayyenn)$
 built from objects in weak $\ell_p$; in detail:\\
 (i) $n/N \goto \delta$;\\
  (ii) $\bx_0 \in \cX_p^w(\xi)$;\\
   (iii) $\bz \in \cZ_2(\sigma)$, and\\
    (iv) $\ayyenn \in \Gauss(n,N)$.
  \eitem
\end{definition}

%
%*********************************************************
%
\subsection{Scalar Minimax Thresholding under Weak $p$-th Moment Constraints}

The class of probability distributions corresponding to instances
in the weak-$\ell_p$ problem suite is
\begin{eqnarray}
\cF^w_p(\xi) \equiv\Big\{\nu\in\cP(\reals)\, :\, 
\sup_{t\ge 0}  t^p \cdot \nu(\{ |X| \geq t \}) \leq  \xi^p
\Big\}\, .
\label{eq:WeakSparseProbs}
\end{eqnarray}
In particular, given a sequence $\bx_0 \in \cX^w_p(\xi)$ ,
the empirical distribution of each $\exxohenn$ is in $\cF^w_p(\xi)$.

As in section \ref{sec:Scalar}, we denote by
$\stMSE(\sigma^2;\nu,\th)$ the mean square error of 
scalar soft thresholding
for a given signal distribution $\nu$.

\begin{definition}
The
{\bf minimax mean squared error}  under the weak $p$-th moment constraint is 
\begin{eqnarray}
M_p^w(\xi) = \inf_{\th\in\reals_+} \sup_{\nu\in\effpee^w(\xi)} 
\E\big\{ \big[ \eta(X + Z;\th) - X \big]^2 \big\}\, ,\label{eq:WeakMpDef}
\end{eqnarray}
where the expectation on the right hand side is taken with respect to
$X\sim \nu$ and $Z\sim\normal(0,1)$, $X$ and $Z$  independent.
\end{definition}

The collection of probability measures $\effpee^w(\xi)$ has a 
distinguished element -- a most dispersed one.
In fact define the envelope function
\begin{eqnarray}
      H_{p,\xi}(t) = \inf_{\nu\in \effpee^w(\xi)} \nu( \{|X|\le t\}) \, ;
\end{eqnarray}
 the envelope of achievable dispersion of the probability mass 
for elements of $\effpee^w(\xi)$.
This envelope can be computed explicitly, yielding 
\begin{eqnarray}
H_{p,\xi}(t) =   \left \{ \begin{array}{ll} 
                                          0                     & 
\mbox{for $t < \xi$,} \\
                                          1 - (\xi/ t)^{p} & 
\mbox{for $t \geq \xi$.} 
                                     \end{array} \right.
\end{eqnarray}
Indeed it is clear by definition that $\nu( \{|X|\le t\})\ge H_{p,\xi}(t)$.
Further defining the CDF
\begin{eqnarray}
F_{p,\xi}^w(x) =   \left \{ \begin{array}{ll} 
                                          \frac{1}{2} + \frac{1}{2} H_p(|x|)    & \mbox{for $x \geq 0$}, \\
                                                                \frac{1}{2} H_p(|x|)    & \mbox{for $x < 0$}.
                                     \end{array} \right .  \label{eq:LeastWeak}
\end{eqnarray}
and letting $\nu_{p,\xi}$ be the corresponding measure,   
we get for any $t\ge 0$,  $\nu_{p,\xi}( \{|X|\le t\}) 
= F_{p,\xi}^w(t)- F_{p,\xi}^w(-t)=H_{p,\xi}(t)$.
We therefore proved the following.
\begin{lemma}
The most dispersed symmetric  probability measure  in $\effpee^w(\xi)$
is $\nu_{p,\xi}$.
This distribution achieves the equality $\nu_{p,\xi}(\{ |X|  \leq t\}) 
= H_{p,\xi}(t)\le \nu(\{ |X|  \leq t\}) $ for all $\nu\in \effpee^w(\xi)$, and
 all $t > 0$.
 \end{lemma}

It turns out that this most dispersed distribution is also the least
favorable distribution for soft thresholding. 
In order to see this fact, define the function $\m:\reals\times\reals_+\to
\reals$ by letting
\begin{eqnarray}
\m(x;\th) = \E\big\{ [\eta(x + Z;\th) - x ]^2\big\}\, ,
\end{eqnarray}
whereby expectation is taken with respect to $Z\sim\normal(0,1)$.
We then have the following useful calculus lemma (see for instance \cite{DMM09}.
 \begin{lemma} \label{lem:monotoneMSE}
For each $\th \in [0,\infty)$, the mapping 
 $x \mapsto \m(x;\th)$ is strictly monotone increasing in $x\in [0,\infty)$.
\end{lemma}

Now the mean square error of scalar soft thresholding,
cf. Eq.~(\ref{eq:ScalarMSEDef}),
is given by
\begin{eqnarray}
\stMSE(1;\nu,\th)\equiv\E \m(|X|;\th) \, ,
\end{eqnarray}
where expectation is taken with respect to $X\sim \nu$.
From the above remarks, we obtain immediately the following 
characterization of the minimax problem.
\begin{corollary}[Saddlepoint]
Consider the game against Nature
where the statistician chooses the threshold $\th$,
Nature chooses the distribution $\nu \in \effpee^w(\xi)$, and the 
statistician pays Nature 
an amount equal to the mean square error
$\stMSE(1;\th,\nu)$. 
 
This game has a saddlepoint $(\th_p^w(\xi),\nu_{p,\xi}^w)$,
i.e. a pair satisfying
\begin{eqnarray}\label{eq:WeakSaddle}
\stMSE(1;\th,\nu_{p,\xi}^w) \geq \stMSE(1;\th_p^w(\xi),\nu_{p,\xi}^w) \geq 
\stMSE(1;\th_p^w(\xi),\nu) \qquad \forall  \tau > 0, \; .
\end{eqnarray}
for all $\th\ge 0$, and $\nu\in \effpee^w(\xi)$.
In particular, the least-favorable probability measure 
is $\nu_{p,\xi}^w = \nu_{p,\xi}$, with distribution $F_{p,\xi}$
given in closed form by Eq.~(\ref{eq:LeastWeak}), and we have the
following formula for the soft thresholding minimax risk:
\begin{eqnarray}
M_p^w(\xi) = \inf_{\tau\ge 0} \stMSE(1;\tau,\nu_{p,\xi})\, .
\label{eq:WeakMinimaxDef}
\end{eqnarray}
\end{corollary}

Ordinarily, identifying a saddlepoint requires search over 
two variables, namely the threshold $\tau$ and the distribution 
$\nu$. In the present problem we need only search over one scalar variable,
i.e.  $\tau$. We can further
make explicit the MSE calculation, by noting that, by Eq.~(\ref{eq:LeastWeak}) 
\begin{eqnarray}
 \stMSE(1;\tau,\nu_{p,\xi}^w) =  
p \cdot \xi^p  \int_{\xi}^\infty \m(x;\tau) x^{-p-1}\; \de x\, .
\end{eqnarray}

\begin{figure}
\begin{center}
\includegraphics[height=3in]{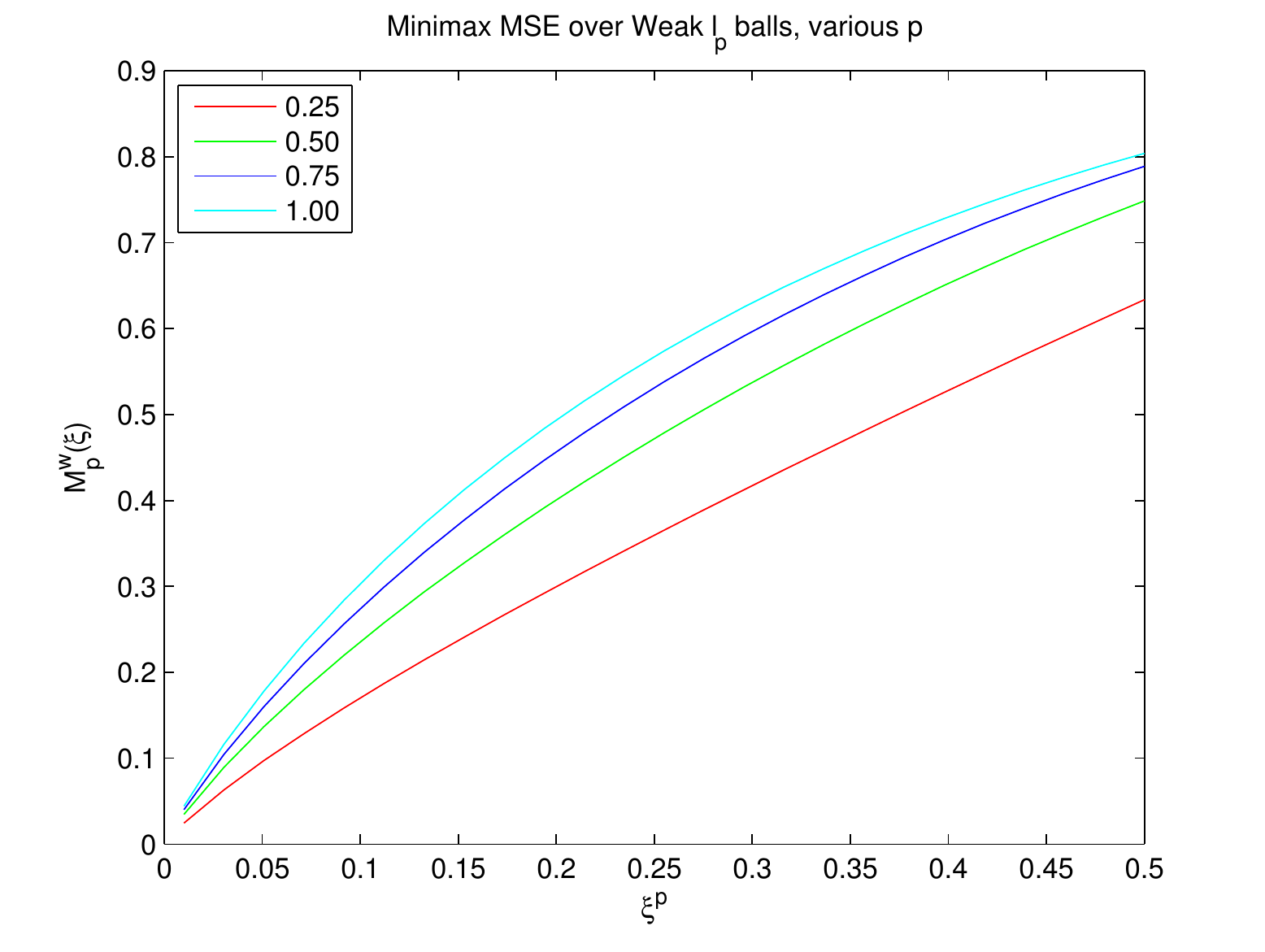}
\caption{Minimax soft thresholding MSE 
over weak-$\ell_p$ balls, $M_p^w(\xi)$, for various $p$.
Vertical axis: worst case MSE over $\cF_p^w(\xi)$. Horizontal axis: $\xi^p$. 
Red, green, blue, aqua curves (from bottom to top)
correspond to $p= 0.25,0.50,0.75,1.00$.}
\label{fig:emmPeeWeak}
\end{center}
\end{figure}

\begin{figure}
\begin{center}
\includegraphics[height=3in]{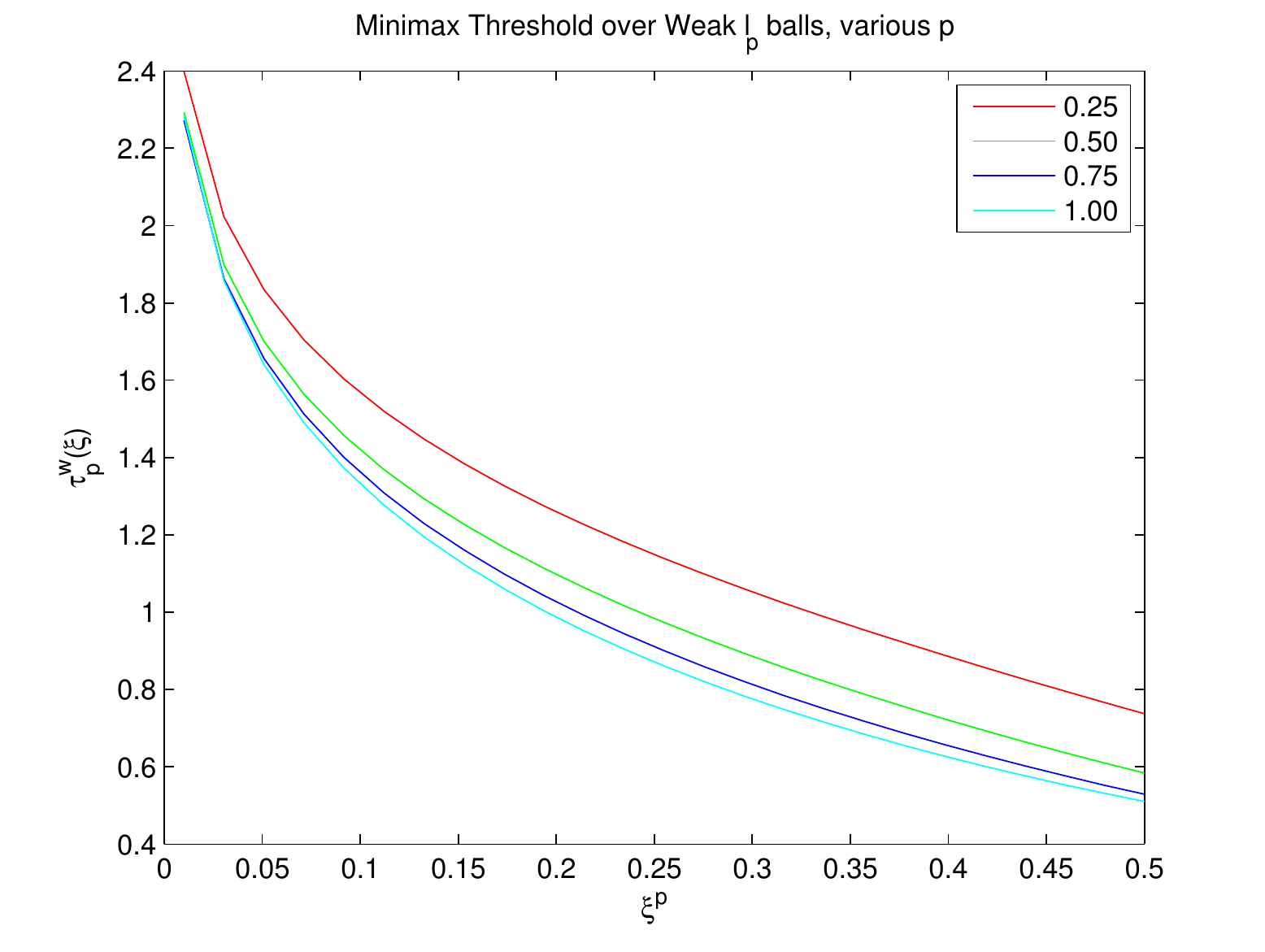}
\caption{Minimax soft threshold parameter, $\tau_p^w(\xi)$, various $p$.
Vertical axis: minimax threshold over $\cF_p^w(\xi)$. Horizontal Axis: $\xi^p$. 
Red, green, blue, aqua curves (from top to bottom)
correspond to $p= 0.25,0.50,0.75,1.00$.}
\label{fig:tauPeeWeak}
\end{center}
\end{figure}

By a simple calculus exercise, this formula and Lemma \ref{lem:monotoneMSE},
imply the following.
\begin{lemma} \label{lem:def:inv:mpw}
The function
$M_p^w(\xi)$ is strictly monotone increasing in $\xi \in (0,\infty)$.
Hence, the inverse function
\[
    (M_p^w)^{-1}(m) = \inf \big\{\, \xi :\; M_p^w(\xi) \geq m \big\} ,
\]
is well-defined for $m \in (0,1)$.
\end{lemma}

The  asymptotic behavior of $M_p^w(\xi)$
in the very sparse limit $\xi\to 0$ was derived  in
\cite{J94}.
\begin{lemma}[\cite{J94}]\label{lemma:AsympWeak}
As $\xi \to 0$, the minimax threshold level
achieving Eq.~(\ref{eq:WeakMinimaxDef})
is given by
\begin{eqnarray*}
\th^w_p(\xi) &=& \sqrt{2 \log( 1/\xi^{p})}\cdot\{1+o(1)\}\, ,\, ,
\end{eqnarray*}
and the corresponding minimax mean square error 
behaves, in the same limit, as 
\beq \label{eq:asymp:MSEWeak}
M^w_p(\xi) = \frac{2}{2-p}\, 
(2 \log( 1/\xi^{p}))^{1-p/2}  \cdot \xi^p\cdot \{1+o(1)\}\, .
\eeq
\end{lemma}
Comparing with Lemma \ref{lem:asymp:scalr:minimax}, we see
that the minimax threshold $\th^w_p(\xi)$ coincides asymptotically with the
one for strong $\ell_p$ balls. The corresponding risk is larger
by a factor $2/(2-p)$ reflecting the larger set of possible distributions
$\nu\in\cF_p^w(\xi)$.   For later use also note:
\[
     (M_p^w)^{-1}(m) = \left( \frac{2-p}{p} m \right)^{1/p} \cdot
                                       \left(  2 \log ( \frac{2-p}{p} m )^{-1} ) \right)^{1/2-1/p} \cdot (1 + o(1) ), \qquad m \goto 0 .
\]

%
%****************************************************************
%
\subsection{Minimax MSE in Compressed Sensing under Weak $p$-th Moments}

We return now to the compressed sensing setup.
In the noiseless case we consider sequences of instances
$\Seq \equiv \{I_{n,N} \} = \{(\exxohenn, \zeeenn=0, \ayyenn)\}_{n,N}$ in $\cS_p^w(\delta,\xi,0)$.
The minimax asymptotic mean square error of the LASSO is then 
given by considering the worst case sequence of instances
\begin{eqnarray}
M_p^{w,*}
(\delta,\xi) \equiv \sup_{\Seq \in \cS_p^{w}(\delta,\xi,0)}  \inf_{\lambda\in\reals_+} 
\AMSE(\lambda;\Seq)\, .
\end{eqnarray}
Here asymptotic mean-square error is defined as
per Eq.~(\ref{eq:AMSE_Def}).

Analogously, in the noisy case $\sigma > 0$, we consider sequences of instances
$\Seq \in \cS_p^w(\delta,\xi,\sigma)$,
We then define the minimax risk as
\begin{eqnarray}
M_p^{w,*}
(\delta,\xi,\sigma) \equiv  \sup_{\Seq \in \cS_p^{w}(\delta,\xi,\sigma)} \inf_{\lambda\in\reals_+} 
\AMSE(\lambda;\Seq)\, .
\end{eqnarray}

It turns out that complete analogs of the results of Sections  \ref{sec:MM_Noiseless}
and \ref{sec:MM_Noisy}
hold for the weak $p$-th moment setting.
Since the proofs are easy modifications of the ones for strong 
$\ell_p$ balls, we omit them.

\begin{theorem}[Noiseless Case, Weak $p$-th moment]
\label{thm:weakellp:noiseless} 
For $\delta \in(0,1)$, $\xi > 0$, the Minimax AMSE 
of the LASSO over the weak-$\ell_p$ ball of radius $\xi$ is:
\beq \label{eq:weakellp:nonoise:minmax:eval}
    M_p^{w,*}(\delta,\xi) = \frac{\delta \xi^2}{(M_p^w)^{-1}(\delta)^2}
\eeq
where $(M_p^w)^{-1}(\delta)$ is the inverse function of the 
soft thresholding minimax risk, see Eq.~(\ref{eq:WeakMinimaxDef}).

Further we have:

\noindent{\bf Least Favorable $\nu$.}  The least-favorable distribution 
$\nu^{w,*}$ is 
the most dispersed distribution $\nu_{p,\xi}$ whose distribution function
is given by Eq.~(\ref{eq:LeastWeak}), with $\xi = \xi^*$.

\noindent{\bf Minimax Threshold.} The minimax threshold $\lambda^{w,*}
(\delta,\xi)$
is given by the calibration relation (\ref{eq:Calibration}) 
with $\th = \th^{w,*}(\delta,\xi)$ determined by:
\begin{eqnarray}
     \th^{w,*}(\delta,\xi) = \th_p^w((M^w_p)^{-1}(\delta)) \, ,
\end{eqnarray}
where $\th_p^w(\, \cdot\,)$ is the soft thresholding minimax
threshold, achieving the infimum in Eq.~(\ref{eq:WeakMinimaxDef}).

\noindent{\bf Saddlepoint.}  The pair
$(\lambda^{w,*},\nu^{w,*})$ satisfies a saddlepoint relation.
Put for short $\AMSE(\lambda; \nu) = \AMSE(\lambda; \delta,\nu,\sigma=0)$.
The minimax AMSE is given by 
\[
M^{w,*}_p(\delta;\xi) = \AMSE(\lambda^{w,*}; \nu^{w,*}),
\]
and
\begin{eqnarray}
 \AMSE(\lambda^{w,*}; \nu^{w,*}) &\leq&  \AMSE(\lambda; \nu^{w,*})\, ,
\qquad \forall \lambda > 0 \\
              & \geq&   \AMSE(\lambda^{w,*}; \nu)
\, , \qquad \forall \nu \in \cF_p^w(\xi) . 
\end{eqnarray}
\end{theorem}
As an illustration of this theorem, consider again the limit
$\delta= n/N\to 0$ after $N\to\infty$ (equivalently, $n/N\to 0$ sufficiently
slowly).  It follows from Eq.~(\ref{eq:asymp:MSEWeak})
that
\begin{eqnarray*}
      M_p^{w,*}(\delta,1) = \Big(1-\frac{p}{2}\Big)^{-2/p}
\delta^{1-2/p} (2\log(\delta^{-1}))^{2/p-1}\big\{1+o_{\delta}(1)\big\}, 
\qquad \delta \goto 0.
\end{eqnarray*}
%
%Rewriting this result explicitly in terms of 
%$\ell_2$ distance, for signals with weak-$\ell_p$ norm bounded by 1,
%we get for $n/N \goto 0$
%%
%\begin{eqnarray}
%%
%\min_\lambda \max_{x_0:\;\|x_0\|_{w\ell_p}^p\le 1}  \|\hx_{\lambda}-x_0\|^2   &=& \Big(1-\frac{p}{2}\Big)^{-2/p}\,  \left(\frac{2  \log(N/n)}{n}\right)^{2/p-1} 
%\big\{1+o_N(1)\big\}\,.  \label{eq:asympMSE:WeakEllpee}
%%
%\end{eqnarray}
%%

We can also compute the minimax regularization parameter.
Lemma \ref{lemma:AsympWeak} gives
\begin{eqnarray}
\lambda^{w,*}(\delta,\xi) = \xi\,\cdot\Big(1-\frac{p}{2}\Big)^{-1/p}\cdot\Big(
\frac{2\log (1/\delta)}{\delta}\Big)^{1/p}\,\big\{1+o(1)\big\}\, , \quad \delta \to 0.
\end{eqnarray}
%
%We can express $\delta$  and $\xi$ in terms of the problem dimension.
%Considering the weak-$\ell_p$ ball of unit radius,
%i.e. the constaint  $\|x_0\|_{w\ell_p}\le 1$, which corresponds
%to $\xi = N^{-1/p}$
%when $n/N\to 0$ slowly after $N,n\to\infty$
%%
%\begin{eqnarray}
%%
%\lambda^{w,*} = \Big(1-\frac{p}{2}\Big)^{-1/p}
%\cdot \left(\frac{2  \log(N/n)}{n}\right)^{1/p} 
%\big\{1+o_N(1)\big\},\;\;\;\;\; \|x_0\|_{w\ell_p}\le 1, n/N \goto 0.  \label{eq:asympLambdaWeak}
%%
%\end{eqnarray}
%%

In the noisy case, we get a result in many respects
similar to the $p$th moment result.
\begin{theorem}[Noisy Case, Weak $p$-th moment] \label{thm:weakellp:noisy}
For any $\delta,\xi>0$, let $m^* = m_p^{w,*}(\delta,\xi)$ be the unique positive 
solution of
\begin{eqnarray}
\frac{m^*}{1+m^*/\delta} = M^w_p\left(\frac{\xi}{(1+m^*/\delta)^{1/2}}\right)\, .
\label{eq:WeakNoisyEq}
\end{eqnarray}
Then the LASSO minimax mean square error $M_p^{w,*}$  is given by:
\beq 
\label{eq:ellp:minmax:evalWeak}
M_p^{w,*}(\delta,\xi,\sigma) =  \sigma^2m^{w,*}_p(\delta,\xi/\sigma)\, . 
\eeq
Further, denoting by $\xi^*\equiv(1+m^*/\delta)^{-1/2}\xi/\sigma$,
we have:

\noindent{\bf Least Favorable $\nu$.} The least-favorable distribution 
$\nu^{w,*}$ is 
the most dispersed distribution $\nu_{p,\xi}$ whose distribution function
is given by Eq.~(\ref{eq:LeastWeak}).

\noindent{\bf Minimax Threshold.} The minimax threshold 
$\lambda^*(\delta,\xi,\sigma)$ is given by the calibration relation 
(\ref{eq:Calibration}) 
with $\th = \th^*(\delta,\xi,\sigma)$ determined as follows:
\begin{align}
\th^{w.*}(\delta,\xi,\sigma) = \th^w_p(\xi^*) .
\end{align}
where $\th_p^w(\, \cdot\,)$ is the soft thresholding minimax
threshold, achieving the infimum in Eq.~(\ref{eq:WeakMinimaxDef}).

\noindent {\bf Saddlepoint.}  The above quantities obey a saddlepoint relation.
Put for short $\AMSE(\lambda;\nu) = \AMSE(\lambda; \delta,\nu,\sigma)$.
The minimax AMSE obeys
\[
M^{w,*}_p(\delta,\xi) = \AMSE(\lambda^{w,*}; \nu^{w,*})\, ,
\]
and
\begin{eqnarray}
 \AMSE(\lambda^*; \nu^*) &\leq&  \AMSE(\lambda; \nu^{w,*})\, ,
\qquad \forall \lambda > 0 \\
              & \geq&   \AMSE(\lambda^{w,*}; \nu)
\, , \qquad \forall \nu \in \cS^w_p(\xi) . 
\end{eqnarray}
\end{theorem}

%
%****************************************************************
%

\section{Traditionally-scaled $\ell_p$-norm Constraints}
\label{sec:Traditional}

This paper uses a non-traditional scaling   $\|x_0\|_p^p \leq N \cdot \xi^p$ for the radius of $\ell_p$ balls;
traditional scaling would be $\|x_0\|_p^p \leq \xi^p$.
In this section we discuss the translation between the two types of conditions.
We first define sequence classes based on norm constraints.
\begin{definition}
The {\bf traditionally-scaled $\ell_p$ problem suite $\tilde{\cS}_p(\delta,\xi,0)$}  is the class 
of sequences of problem instances $I_{n,N} = (x_0^{(N)}, z^{(n)},A^{(n,N)})$ where:\\
 (1) $n/N \goto \delta$; \\
 (2) $\|x_0^{(N)}\|_p^p \leq \xi^p$, and, for some sequence
 $B= \{B_M\}_{M\ge 0}$ such that
$B_M\to 0$, we have
$\sum_{i=1}^N(x_{0,i}^{(N)})^2\ind(|x_{0,i}^{(N)}|\ge M) \leq B_M N^{1-2/p}$ for every $N$;  \\
(3)  $z^{(n)} \in \bR^n$,  $\| z^{(n)} \|_2  \sim \sigma \cdot n^{1/2} \cdot N^{-1/p} $, $(n,N) \goto \infty$. \\
(4) $A^{(n,N)} \sim \Gauss(n,N)$.

The {\bf traditionally-scaled weak $\ell_p$ problem suite $\tilde{\cS}_p^w(\delta,\xi,0)$} is defined  using conditions (1),(3),(4) and \\
$(2^w)$ $\|x_0^{(N)}\|_{w\ell_p}^p \leq \xi^p$, and, , for some sequence
 $B= \{B_M\}_{M\ge 0}$ such that
$B_M\to 0$, we have
$\sum_{i=1}^N(x_{0,i}^{(N)})^2\ind(|x_{0,i}^{(N)}|\ge M) \leq B_M
N^{1-2/p}$ for every $N$;
\end{definition}

Comparing our earlier definitions of standard $\ell_p$-constrained problem suites $\cS_p(\delta,\xi,\sigma)$
and $\cS_p^w(\delta,\xi,\sigma)$ with these new definitions,
conditions (1) and (4) are identical; while the new (2) and (3) are simply rescaled
versions of corresponding conditions (2) and (3) in the earlier standard problem suites.\footnote{Note the awkwardness of the noise scaling in the traditional scaling, as compared to
the standard scaling used here}
To deal with such rescaling,
we need the following scale covariance property:

\begin{lemma}
Let $I = (x_0^{(N)}, z^{(n)},A^{(n,N)})$ be a problem instance and $I_a = (a \cdot x_0^{(N)}, a \cdot z^{(n)},A^{(n,N)})$
be the corresponding dilated problem instance.
Suppose that $\hat{x}_\lambda^{(N)}$ is the unique LASSO solution generated by instance $I$
and $\hat{x}_{\lambda}^{(N),a}$ the unique solution generated by instance $I_a$. Then
\[
    \hat{x}_{a\lambda}^{(N),a} = a \cdot  \hat{x}_{\lambda}^{(N)},
\]
\[
     \|  \hat{x}_{a\lambda}^{(N),a}  - a x_0^{(N)} \|_2^2 =  a^2 \cdot  \|  \hat{x}_{\lambda}^{(N)}  -  x_0^{(N)} \|_2^2, 
\]
and
\[
    \inf_\lambda  E \|  \hat{x}_{\lambda}^{(N),a}  - a x_0^{(N)} \|_2^2 =  a^2 \cdot  \inf_\lambda  E \|  \hat{x}_{\lambda}^{(N)}  -  x_0^{(N)} \|_2^2 .
\]
\end{lemma}

Applying this lemma yields the following problem equivalences:

\begin{corollary}
We have the scaling relations:
\[
    \sup_{\Seq \in \tilde{\cS}_p(\delta,\xi,0)}   \inf_\lambda \AMSE(\lambda, \Seq) =   N^{-2/p} \cdot  \sup_{\Seq \in {\cS}_p(\delta,\xi,0)}   \inf_\lambda \AMSE(\lambda,\Seq)  ;
\]
and
\[
    \sup_{\Seq \in \tilde{\cS}_p^w(\delta,\xi,0)}   \inf_\lambda \AMSE(\lambda, \Seq)=   N^{-2/p} \cdot  \sup_{\Seq \in {\cS}_p^w(\delta,\xi,0)}   \inf_\lambda \AMSE(\lambda, \Seq)  .
\]
\end{corollary}

Let's  apply this to noiseless $\ell_p$ ball constraint.  By Theorem 4.1 we have
\[
 \min_\lambda \max_{\Seq \in \cS(\delta,\xi,0)}    \AMSE(\lambda ,\Seq ) =    \frac{\delta \xi^2}{M_p^{-1}(\delta)^2}
\]
Considering the {\it un}normalized squared error $\|\hx_{\lambda}-x_0\|^2$ and operating purely formally,   define
a symbol $\bar{E}$ so that when
$x_0^{(N)}$ arises from 
a given sequence   $\Seq$,
\[
     \bar{E} \|\hx_{\lambda}^{(N)}-x_0^{(N)}\|^2  = N \cdot \AMSE(\lambda,\Seq).
\]
Remembering $\delta = n/N$ we have
\begin{eqnarray*}
 \min_\lambda \max_{\Seq \in{\cS}_p(\delta,\xi,0)}    \bar{E}  \|\hx_{\lambda}^{(N)}-x_0^{(N)}\|^2 &=& 
N \cdot (n/N) ^{1-2/p} \cdot  \xi^2 \cdot  \left(2  \log(N/n)\right)^{2/p-1} 
\big\{1+o_N(1)\big\}\, . \\
&=&       N^{2/p} \xi^2 \cdot  \left(\frac{2  \log(N/n)}{n}\right)^{2/p-1}  \big\{1+o_N(1)\big\} .
\end{eqnarray*}
Using the traditionally-scaled $\ell_p$ problem suite, 
\[
 \min_\lambda \max_{\Seq \in \tilde{\cS}_p(\delta,\xi,0)}    \bar{E}  \|\hx_{\lambda}-x_0\|^2 =
     N^{-2/p}  \cdot   \min_\lambda \max_{\Seq \in{\cS}_p(\delta,\xi,0)}    \bar{E}  \|\hx_{\lambda}-x_0\|^2 ,
\]
where on the LHS we have $\tilde{\cS}_p(\delta,\xi,0)$ while on the RHS we have ${\cS}_p(\delta,\xi,0)$.
We conclude 
\begin{corollary}  Consider the noiseless, traditionally-scaled $\ell_p$ problem formulation. The asymptotic MSE
for the $\ell^2$-norm error measure has the asymptotic form
\begin{equation}   \label{eq:asympMSE:Ellpee}
    \min_\lambda \max_{\Seq \in \tilde{\cS}_p(\delta,\xi,0)}    \bar{E}  \|\hx_{\lambda}-x_0\|^2 = \xi^2 \cdot  \left(\frac{2  \log(N/n)}{n}\right)^{2/p-1}  \big\{1+o_N(1)\big\} ;
\end{equation}
this is valid both for $n/N \goto \delta \in (0,1)$ and for $\delta = n/N \to 0$ slowly enough.
The maximin penalization has an elegant prescription when $n/N \to 0$ slowly enough:
\begin{eqnarray}
\lambda^* =  \xi \cdot \left(\frac{2  \log(N/n)}{n}\right)^{1/p}, \qquad \Seq \in \tilde{\cS}_p(\delta,\xi,0)   \label{eq:asympLambda}.
\end{eqnarray}
\end{corollary}

Our results can now be compared with  earlier results written in the traditional scaling.
We rewrite our result for $\xi=1$, using a simple moment condition
that implies uniform integrability. For  all sufficiently large $B$, and
all $q>2$, we obtained:
\begin{equation} \label{eq:thisPaper:ellpee}
    \min_\lambda \max_{
                        \|x_0\|_p^p  \le 1\, ,\; \|x_0\|^q_q\le BN^{1-q/p}
                                                       }    \bar{E}  \|\hx_{\lambda}^{(N)}-x_0^{(N)}\|^2 =  \left(\frac{2  \log(N/n)}{n}\right)^{2/p-1}  \big\{1+o_N(1)\big\} .
\end{equation}
In the case $\lambda=0$, earlier results  \cite{Donoho1,CandesTao} imply:
\begin{eqnarray}
\max_{\|x_0\|_p^p\le 1}  \|\hx_{0}-x_0\|^2   &=&   O_P\left( \Big( \frac{ \log(N/n)}{n}\Big)^{2/p-1}  \right) ,\;\;\;\;\; n/N \goto 0.
 \label{eq:thosePapers:ellpee}
\end{eqnarray}
There are two main differences in {\it technical content} between the new result and earlier ones
\bitem
 \item The use of $\bar{E}$ on the LHS of (\ref{eq:thisPaper:ellpee}) versus $O_P(\,\cdot\,)$ on the RHS of (\ref{eq:thosePapers:ellpee}).
 \item The supremum over $\{\|x_0^{(N)}\|_p^p\le 1\}$ on the LHS  of (\ref{eq:thosePapers:ellpee}) versus
 the supremum  over $\{\|x_0^{(N)}\|_p^p\le 1\} \cap \{ \|x_0\|_{q}^q\leq B N^{1-q/p}\}$ on the LHS  of (\ref{eq:thisPaper:ellpee}).
 \eitem
 The main difference in {\it results} is of course that the new result gives a precise constant in place of the $O(\,\cdot\,)$
 result which was previously known. See Section \ref{sec:ComparisonWidth} for 
further discussion.

The new result has the additional ingredient, not seen earlier, that we constrain not only  $\{\|x_0^{(N)}\|_p^p\le 1\}$ 
but also $ \{ \|x_0^{(N)}\|_2^2 \leq B N^{1-q/p}\}$.  For each  $p < 2$, this additional constraint does indeed give a smaller
set of feasible vectors for large $N$.  See Section \ref{sec:LiteratureSection} for 
further discussion.

A traditionally-scaled weak-$\ell_p$ problem suite  $\tilde{\cS}_p^w(\delta,\xi,\sigma)$ can also be defined; without giving details, we have:
\begin{corollary}  Consider the noiseless, traditionally-scaled weak-$\ell_p$ problem formulation. The asymptotic MSE
for the $\ell^2$-norm error measure has the asymptotic form
\begin{equation} \label{eq:asympMSE:WeakEllpee}
    \min_\lambda \max_{\Seq \in \tilde{\cS}_p^w(\delta,\xi,0)}    \bar{E}  \|\hx_{\lambda}-x_0\|^2 =  (1 - p/2)^{-2/p} \cdot  \xi^2 \cdot  \left(\frac{2  \log(N/n)}{n}\right)^{2/p-1}  \big\{1+o_N(1)\big\} ;
\end{equation}
this is valid both for $n/N \goto \delta \in (0,1)$ and for $\delta = n/N \to 0$ slowly enough.
The maximin penalization has an elegant prescription for $n/N$ small:
\begin{eqnarray}
\lambda^* =  (1 - p/2)^{-1/p} \cdot  \xi \cdot \left(\frac{2  \log(N/n)}{n}\right)^{1/p}, \qquad \Seq \in \tilde{\cS}_p^w(\delta,\xi,0)   \label{eq:asympLambdaWeak}.
\end{eqnarray}
\end{corollary}

\section{Compressed Sensing over the Bump Algebra}

Our discussion involving $\ellpee $-balls 
is so far rather abstract. We consider here  a stylized application: 
recovering a signal $f$ in the Bump Algebra
from compressed measurements. Consider a function $f:[0,1]\to\reals$
which admits the  representation
\begin{eqnarray}
    f(t) = \sum_{i=1}^{\infty}  c_i\, g\big( (t - t_i)/\sigma_i \big)\, ,
\;\;\;\;\;\;\;\; g(x) = \exp\big(-x^2/2\big), \quad \,  \sigma_i > 0 .\label{eq:Bump}
\end{eqnarray}
Each term $g(\,\cdot\,)$ is a  Gaussian `bump' normalized to
height 1, and  we assume $\sum_{i=1}^{\infty} |c_i| \leq 1$
which ensures convergence of the series.
The $c_i$ are signed amplitudes of the `bumps' in $f$.
We refer to the book by Yves Meyer \cite{Meyer}
and also to the discussion in \cite{DJ98}, which calls such objects
models of {\it polarized spectra}.
Any such function also has a wavelet representation
\[
   f = \sum_{j \geq -1} \sum_{k\in \cI_j} \alpha_{j,k} \psi_{j,k} ,
\]
where the $\psi_{j,k}$ are smooth orthonormal wavelets (for example
Meyer wavelets or Daubechies wavelets), and the wavelet coefficients
obey $\sum_{j,k} |\alpha_{j,k}| \leq C$. The constant $C$ depends only 
on the wavelet basis \cite{Meyer}. 
Here $j$ denotes the level index, and $k$ the position index.
We have $|\cI_{-1}| = 1$, and $|\cI_j| = 2^{j}$ for each $j\ge 0$.
In other words the collection of functions with wavelet coefficients
in an $\ell_1$-ball of radius $C$ contains the whole algebra
of functions representable as in 
(\ref{eq:Bump}).

Now consider compressed sensing of such an object.
We fix a maximum resolution, by picking $N = 2^J$ and considering the 
finite-dimensional problem of recovering the object 
$f_N = \sum_{j <  J} \sum_{k=0}^{2^j-1} \alpha_{j,k} \psi_{j,k}$.
The  scale $2^{-J}$ corresponds to an effective discretization
scale: on intervals of length much smaller than $2^{-J}$,
the function $f_N$ is approximately constant.
Reconstructing the function $f_N$ is equivalent to
recovering the $2^J$ coefficients
 {\footnotesize
\[
x_0 = \left ( \alpha_{-1,0},\alpha_{0,0}, \alpha_{1,0},\alpha_{1,1}, \alpha_{2,0}, \dots ,\alpha_{2,3},\alpha_{3,1},\dots,  \alpha_{J-1,0},\dots,\alpha_{J-1,2^{J-1}-1} \right ) .
\]
}
We know that coefficients at scales $1$ through $J-1$ combined
have a total $\ell_1$-norm bounded by a numerical constant
$C$. Without loss
of generality, we shall take $C=1$ (this corresponds to 
rescaling the constraint on the bump representation
(\ref{eq:Bump})).

Denote by $V_J$ the $2^{J}$-dimensional space of functions on 
$[0,1]$, with resolution $2^{-J}$, i.e. 
\begin{eqnarray}
V_J \equiv \Big\{\sum_{j <  J} \sum_{k=0}^{2^j-1} \alpha_{j,k} \psi_{j,k}: 
\; \; \alpha_{j,k}\in\reals \Big\}\, .
\end{eqnarray}
We can construct a random linear \emph{measurement operator} 
$\cA: V_J\to \reals^n$, such that the matrix $A$ representing $\cA$ in 
the basis of wavelets has random Gaussian coefficients iid $\normal(0,1)$.  
We then take $n+1$ noiseless measurements: the scalar 
$\alpha_{-1,0} = \langle f , \psi_{-1,0} \rangle$
associated to the `father wavelet', and the vector
$y = \cA f_N$. Notice that, since the measurements are noiseless, 
the variance of the entries of the measurement matrix $A$ 
can be rescaled arbitrarily.

In the wavelet basis, the measurements can be rewritten as
$y = A x_0$, where the $A$ is an $n \times N$ Gaussian random matrix.
This is precisely a problem of the type studied in earlier sections.
Suppose now that we apply $\ell_1$-penalized least-squares 
\begin{eqnarray}
    \hx_{\lambda}\equiv \arg\min_x \Big\{\frac{1}{2}
\|y - Ax\|_2^2  + \lambda \| x \|_1 \Big\}\, ,
\end{eqnarray}
and denote the entries of the reconstruction vector by
$\hx_{\lambda}\equiv (\halpha_{0,0}, \dots , \halpha_{J-1,2^{J-1}-1})$.
The function $f_N$ is therefore reconstructed as 
$\hf_N$, where
\[
   \hf_N =  \sum_{j <  J} \sum_{k=0}^{2^j-1} \hat{\alpha}_{j,k} \psi_{j,k}\,  .
\]
We adopt the performance measure
\[
      \MSE(\hf_N,f_N) \equiv  
\E \big\{\| f_N - \hf_N \|_{L_2[0,1]}^2\big\} 
=  \E \| x_{0}  -\hx_{\lambda} \|_2^2 ;
\]
where the last equality uses the orthonormality of the wavelet basis.

We wish to choose an appropriate value of $\lambda \geq 0$ to 
give the best reconstruction
performance. Note that the coefficients vector $x_0\in\reals^N$ 
satisfies by assumption
\[
  \|x_0\|_1   \leq  1 \,
\]
so we are in the setting of traditionally-scaled $\ell_p$ balls.
The discussion of the last section now applies; we obtain results 
by rescaling results from Theorem 4.1.
 Letting $\lambda_p^*(\delta,\xi)$ denote the
minimax threshold of Theorem 4.1, define
\beq \label{eq:regular}
  \lambda_N =  N^{-1} \cdot   \lambda_1^*( \frac{n}{N},1).
\eeq
%Keep in mind the small-$\delta$ form (\ref{eq:asymp:lambda})-(\ref{eq:asympLambda})
%
\begin{corollary} 
Consider a sequence of functions $f_N\in V_J$ in the Bump Algebra (normed so that the
wavelet coefficients have $\ell_1$-norm bounded by $1$).  Consider
Gaussian measurement operators $\cA_N:V_J\to\reals^n$ 
indexed by the problem dimensions
$N=2^J$, and $n$. 
Let $\hf_N^*$ denote the reconstruction  of $f_N$  using regularization parameter 
$\lambda= \lambda_N$ of  (\ref{eq:regular}).

(i) Assume $n/N\to\delta\in (0,1)$.
Then we have
\begin{eqnarray}
\MSE(\hf_N^*,f_N)  \le   N^{-1} \cdot M_1^*(\delta, 1 )\, (1+o(1))\, ,
\end{eqnarray}
with $M_1^*(\delta,\xi)$ as in Theorem 4.1.  This bound is asymptotically tight (achieved for
a specific sequence $f_N$).

(ii) Assume $n/N\to 0$ sufficiently slowly.
Then we have
\begin{eqnarray}
\MSE(\hf_N^*,f_N)  \le  {\frac{2 \log(N/n)}{n}} \cdot (1+o(1)), \, ,
\end{eqnarray}
and the bound is asymptotically tight (achieved for
a specific sequence $f_N$).
\end{corollary}

\section{Compressed Sensing over Bounded Variation Classes}

Compressed sensing problems make sense for many other
functional classes. The class of Bounded Variation affords an application of
our results on weak $\ell_p$ classes.
\begin{enumerate}
 \item Every bounded variation function 
$f\in BV[0,1]$ has Haar wavelet coefficients in a weak-$\ell_{2/5}$ ball.
 \item Every $f \in BV[0,1]^2$ has wavelet coefficients in a weak-$\ell_1$ 
ball \cite{CohenNonlinear}.
%
%\item %
\end{enumerate}

We can develop a theory of compressed sensing over BV spaces
following the previous section, now using Haar
wavelets. $V_J$ means again the span wavelets of spatial scale $2^{-J}$ or coarser.
We let $d$ denote the spatial dimension ($d=1$ or $2$ in the above examples).
We use regularization parameter
 \beq \label{eq:regular:weakellpee}
  \lambda_N =  N^{-1} \cdot   \lambda^{w,*}(n/N,1).
\eeq

\begin{corollary} 
Consider a sequence of functions $f_N\in V_J$ whose Haar
wavelet coefficients have weak $\ell_p$-norm bounded by $1$.  Consider
Gaussian measurement operators $\cA_N:V_J\to\reals^n$ 
indexed by the problem dimensions
$N=2^{dJ}$, and $n$. 
Let $\hf_N^*$ denote the reconstruction  of $f_N$  using regularization parameter 
$\lambda= \lambda_N$ of  (\ref{eq:regular:weakellpee}).

(i) Assume $n/N\to\delta\in (0,1)$.
Then we have
\begin{eqnarray}
\MSE(\hf_N^*,f_N)  \le   N^{-1} \cdot M_p^{w,*}(\delta, 1 )\, (1+o(1))\, ,
\end{eqnarray}
with $M_1^{*,w}(\delta,\xi)$ as in Theorem 6.1.  This bound is asymptotically tight (achieved for
a specific sequence $f_N$).

(ii) Assume $n/N\to 0$ sufficiently slowly.
Then we have
\begin{eqnarray}
\MSE(\hf_N^*,f_N)  \le \Big(1 -\frac{p}{2}\Big)^{-2/p} \cdot \left( {\frac{2 \log(N/n)}{n}} \right)^{2/p-1} \cdot (1+o(1)), \, ,
\end{eqnarray}
and the bound is asymptotically tight (achieved for
a specific sequence $f_N$).
\end{corollary}

Although BV offers only the applications $p=1$ ($d=2$) and $p=2/5$ ($d=1$),
weak-$\ell_p$ spaces arise elsewhere, and serve as useful models for image content.
For example,  for images containing smooth edges,
we have the following model: every  $f : [0,1]^2 \mapsto \reals$ which is locally in $C^2$ except at $C^2$ `edges' has 
curvelet coefficients levelwise in weak-$\ell_{2/3}$
balls \cite{CandesDonohoTight}.  Our compressed sensing  result for BV
can be adapted without change to the conclusions for such a setting, after replacing
the role of Haar wavelets by Curvelets.

%
%*********************************************
%
\section{Discussion}
\label{sec-discuss}

In this last section we discuss some specific aspects of our results and
overview (in an unavoidably incomplete way) the related literature.

\subsection{Equivalence of Random and Deterministic Signals/Noises}
\label{sec:RandomVsDeterministic}

A striking aspect of our results
is the equivalence of random and deterministic signals
and noises  (traceable here to Proposition 3.1).  
The AMSE formula in the general case, as given by 
Eq.~(\ref{eq:asymptotic-result}), depends on the sequence of
signals $\exxohenn$ and of noise vectors $\zeeenn$ only through simple 
statistics of such vectors. 
More precisely, it depends only on their asymptotic empirical distributions,
respectively $\nu$ and $\omega$. 
In fact the dependence on $\zeeenn$ is even weaker:
the asymptotic risk only depends on the limit second moment 
$\E_{\omega}(Z^2)$.

At first sight, these findings are somewhat surprising.
For instance we might replace $\exxohenn$ with a random vector with 
i.i.d. entries with common distribution $\nu$ without changing 
the asymptotic risk. This asymptotic equivalence between 
random and deterministic signal is in fact a quite
simple and robust consequence of the absence of structure
of the measurement matrix $A$. 
We do not spell out the details here,
but note the following simple facts
\begin{enumerate}
 \item  Under our model for $A$, the columns of $A$ are exchangeable, 
so there is no distributional
difference between $Ax_0$ and $A P x_0$, for any permutation matrix $P$.
\item  As a consequence, there is no difference in expected
performance between a fixed vector $x_0$ and a random vector obtained 
by permuting the entries of $x_0$ uniformly at random.
\item Asymptotically for large $N$ there is a negligible difference 
in performance
 between a fixed vector $x_0$ and the typical random vector
obtained by sampling with replacement from the entries of $x_0$.
\end{enumerate}
This argument implies that we can replace the deterministic
vectors $\exxohenn$ with random vectors with i.i.d. entries. As the 
argument clarifies, this phenomenon ought to exist for more general 
models of $A$. 
%
%*******************************************************
%
\subsection{Comparison with Previous Approaches}
\label{sec:LiteratureSection}

Much of the analysis of compressed sensing 
reconstruction methods  has relied so far on a kind of {\it qualitative analysis}.
A typical approach has been to frame the analysis in terms  of 
`worst case' conditions on the measurement matrix 
$A$.  A useful set of conditionsis provided by the restricted
isometry property (RIP), \cite{CandesTao,CandesStable}
 and refinements  \cite{BickelEtAl,BuhlmannLASSO,Indyk}.
 These conditions are typically pessimistic, in that they assume that
 the signal $x_0$ is chosen adversarially, but they 
capture the  correct scaling behavior.

The advantage of this approach is its broad applicability;
since one assumes little about the matrix $A$, the derived bound 
will perhaps apply to a wide range of matrices. However,
there are two limitations:
\begin{enumerate}
\item[(a)] These  conditions have been 
proved to hold with linear scaling of $\|x_0\|$
and $n$ with the signal dimension $N$, only
for specific random ensembles of measurement matrices,
e.g. random matrices with i.i.d. subexponential entries.
\item[(b)] The resulting bounds typically only hold up to unspecified numerical constants. 
Efforts to make precise the implied constants in specific cases (see for instance \cite{BlCaTa09})
show  that this approach imposes restrictive conditions
on the signal sparsity.
For instance, for a Gaussian measurement matrix with undersampling
ratio $\delta=0.1$, RIP implies successful reconstruction 
\cite{BlCaTa09} only if $\|x_0\|_0\lesssim 0.0002\,N$.
In empirical studies, a much larger support appears to be tolerated.
\end{enumerate}

The present paper works with only one matrix ensemble -- Gaussian random matrices -- 
but gets quantitatively precise results, like the companion works
 \cite{DMM09,NSPT,BayatiMontanariLASSO}.
The approach provides sharp performance guarantees under 
suitable probabilistic models for the measurement process.

To be concrete, consider the case of $x_0$ belonging to the weak-$\ell_p$
ball of radius $1$, $\|x_0\|_{w\ell_p}\le 1$.
Building on the RIP theory, the 
review paper \cite{CandesReview} derives the bound
\begin{eqnarray}
  \|\hx_{\lambda}-x_0\|^2   &\le & C
\, \left(\frac{  \log(N/n)}{n}\right)^{2/p-1} \, ,
\end{eqnarray}
holding for Gaussian measurement matrices $A$, and for unspecified
constant $C$. Analogous minimax bounds for $\ell_p$
balls are known \cite{Donoho1,WainwrightEllP}.
Our  results have the same form, but with 
specific constants, e.g. $C=(1-(p/2))^{-2/p}$ for
weak-$\ell_p$ balls, cf. Eq.~(\ref{eq:asympMSE:WeakEllpee})
and $C=1$ for ordinary $\ell_p$ balls, cf. Eq.~(\ref{eq:asympMSE:Ellpee}).
Moreover, these constants are sharp, i.e. attained by specifically described $x_0$.

Let us finally mention the recent paper \cite{CandesPlanRIPless},
that takes a probabilistic point of view similar 
to the one of \cite{NSPT} and to the present one,
although using different techniques. This approach
avoids using RIP or similar conditions, and applies
to a broad family of matrices with i.i.d. rows. On the other hand,
it only allows to prove  upper bounds on MSE
off by logarithmic factors.

%
%*******************************************************
%
\subsection{Comparison to the theory of widths}
\label{sec:ComparisonWidth}

Recall that the Gel'fand $n$-width of of a set $K\subseteq \reals^N$
with respect to the norm $\|\,\cdot\, \|_X$  is defined as 
\begin{eqnarray}
d_n(K,X)  = \inf_{A\in\reals^{n\times N}}\sup_{x\in K\cap\ker(A)}\|x\|_X\, ,
\end{eqnarray}
where $\ker(A)\equiv \{v\in\reals^N:\, Av=0\}$. 
Here we shall consider $K$ to be the $\ell_p$ ball
of radius $1$, $B_{p}^N \equiv \{x\in\reals^N\, :\; \|x\|_p \le 1\}$,
and fix $\|\,\cdot\, \|_X$ to be the ordinary $\ell_2$ norm.
A series of works \cite{Kashin1977,GarnaevGluskin,Donoho1,FoucartEtAl} 
established that
\begin{eqnarray}
d_n(B_p^N,\ell_2)  \ge c_p\left(\frac{ \log(N/n)}{n}\right)^{1/p-1/2}
\label{eq:Width}
\end{eqnarray}
as long as the term in parenthesis is smaller than $1$.

The interest for us lies in the well-known observation \cite{Donoho1}
that $d_n(B_p^N,\ell_2)$ provides a lower bound on the compressed
sensing mean square error under arbitrary reconstruction algorithm,
and for arbitrary measurement matrix $A$.
In particular
\begin{eqnarray}
 \max_{x_0\in B_p^N}  \|\hx-x_0\| 
 \ge d_n(B_p^N,\ell_2) \, .
\end{eqnarray}
So it makes sense to define the inefficiency of 
a certain matrix/reconstruction procedure as
the ratio of the two sides in the above inequality
\begin{eqnarray}
 r_{\rm alg}(B_p^N,\ell_2) \equiv \frac{1}{d_n(B_p^N,\ell_2)}\max_{x_0\in B_p^N}  \|\hx-x_0\| \, .
\end{eqnarray}
This ratio implicitly depends on the matrix  $A$.
In the case $p=1$, $\lambda=0$ it is known that $\ell^1$ minimization is inefficient at most by a factor $2$:
\begin{eqnarray}
 r_{{\rm min} \; \ell_1}(B_p^N,\ell_2) \leq  2;
\end{eqnarray}
 (for example \cite{Donoho1} showed this by invoking  
\cite{Traub80}).

Our work concerns random Gaussian matrices and LASSO reconstruction.
Since the worst-case performance of the optimally-tuned LASSO 
can not be  worse than the worst-case performance of min-$\ell^1$ reconstruction,
and since we have a formal expression for the worst-case AMSE of 
optimally-tuned LASSO, the  asymptotic formula (\ref{eq:asympMSE:Ellpee})
together with the bound (\ref{eq:Width})
implies for all sufficiently large $B$ and any $q>2$ that 
\begin{eqnarray}
      \max_{x_0\in B_p^N, \; \|x_0\|_q \leq B
        N^{1/q-1/p}}\overline{E} \|\hx_0^{(N)}-x_0\|  =  \sqrt{ \frac{2
          \log(N/n)}{n}} ( 1+ o(1)),\label{eq:Last}
\end{eqnarray}
with $\overline{E}$ defined in analogy with Section \ref{sec:Traditional}.
The constant $B$ is arbitrary, which suggests (but of course does not prove) that we can 
remove the hypothesis  $\|x_0\|_q \leq B N^{1/q-1/p}$
completely.

On the other hand, for a fixed matrix $A\in\reals^{n\times N}$,  we can define the  width
\[
d_n(K,A,X)  = \sup_{x\in K\cap\ker(A)}\|x\|_X\, ,
\]
so that the Gel'fand $n$-width is the infimum of this quantity over $A$.
Using results of Donoho and Tanner \cite{DoTa08}
one can give the lower bound for $p=1$ and Gaussian random matrices
\[
       d_n(B_p^N,A,\ell_2)  \geq   \sqrt{ \frac{ \log(N/n)}{4 e n}} ( 1+ o(1)).
\]
The right hand side of Eq.~(\ref{eq:Last}) is quantitatively quite
close to  the right-hand side of the last display. 
 Hence the results of this paper suggest 
that  statistical methods may also provide geometric information.

In the general case $0 < p < 1$, lower bounds on $c_p$ are given
in \cite{FoucartEtAl}, but they do not appear as tight as desirable.

\subsection{About the Uniform Integrability Condition}
\label{sec:UniformIntegrability}

We have just seen once again that our hypotheses on $\ell_p$ balls
can be scaled to match $\| x_0 \|_p \leq 1$ but then they also include the 
hypothesis $\|x_0\|_q \leq B N^{1/q-1/p}$.  It may seem at first glance that this
is a serious additional constraint; it implies that the entries in
$x_0$ cannot be very large as $N$ increases, whereas  the condition
$\| x_0 \|_p \leq 1$  of course permits entries as large as 1.  

However, note that our analysis identifies the least-favorable $x_0$,
and that the constant $B$ plays no role.  In fact, if we make a homotopy between
the least-favorable object and objects requiring larger  $B$,  we find that the AMSE is decreasing in the
direction of larger $B$.  Pushing things to the extreme where $B$ goes unbounded, of course
our analysis techniques no longer rigorously apply, but it is quite clear that this is an unpromising
direction to move. Hence we believe that this is largely a technical condition, caused by our method of proof.

%
%**************************************************************
%
\bibliographystyle{amsalpha}
\newcommand{\etalchar}[1]{$^{#1}$}
\providecommand{\bysame}{\leavevmode\hbox to3em{\hrulefill}\thinspace}
\providecommand{\MR}{\relax\ifhmode\unskip\space\fi MR }
\providecommand{\MRhref}[2]{%
  \href{http://www.ams.org/mathscinet-getitem?mr=#1}{#2}
}
\providecommand{\href}[2]{#2}

\end{document}